\documentclass[11pt,a4paper,oneside]{article}
\normalbaselines 
\usepackage{amsmath}
\usepackage{amssymb}
\usepackage{mathtools}
\usepackage{mathcomp}
\usepackage{textcomp}
\usepackage{amsthm}
\usepackage{enumerate}
\usepackage[auth-lg]{authblk}
\usepackage{cite}
\usepackage{tikz}
\usepackage{caption}
\usepackage{subcaption}
\usepackage[ruled,linesnumbered,vlined]{algorithm2e}
\usetikzlibrary{arrows}
\usetikzlibrary{shapes.geometric,patterns, positioning, fit}
\usepackage{multicol}
\usepackage{tikz-qtree}
\usepackage{rotating}
\usepackage{url}
\usepackage{array}
\newcolumntype{L}[1]{>{\raggedright\let\newline\\\arraybackslash\hspace{0pt}}m{#1}}
\newcolumntype{C}[1]{>{\centering\let\newline\\\arraybackslash\hspace{0pt}}m{#1}}
\newcolumntype{R}[1]{>{\raggedleft\let\newline\\\arraybackslash\hspace{0pt}}m{#1}}
\usepackage{floatrow}

\newtheorem{theorem}{Theorem}[section]\theoremstyle{plain}
\newtheorem{corollary}[theorem]{Corollary}\theoremstyle{plain}
\theoremstyle{plain}
\newtheorem{lemma}[theorem]{Lemma}\theoremstyle{plain}
\theoremstyle{plain}
\newtheorem{claim}{Claim}\theoremstyle{plain}
\theoremstyle{plain}
\theoremstyle{plain}
\theoremstyle{plain}
\newtheorem{definition}[theorem]{Definition}\theoremstyle{plain}
\theoremstyle{plain}
\usepackage[margin=1.2in]{geometry}

\newcommand{\SFS}{\text{\rm SFS}}
\newcommand{\DiSFS}{\text{\rm DiSFS}}
\newcommand{\ignore}[1]{}

\newcommand{\MS}{{\mathcal S}}

\newcommand{\MB}{{\mathcal B}}
\newcommand{\MC}{{\mathcal C}}

\begin{document}
\title{\bf Similarity-First Search: a new algorithm with application to Robinsonian matrix recognition}
\author[1,2]{Monique Laurent}
\author[1]{Matteo Seminaroti}
\date{}

\affil[1]{\small Centrum Wiskunde \& Informatica (CWI), Science Park 123, 1098 XG Amsterdam, The Netherlands}
\affil[2]{\small Tilburg University, P.O. Box 90153, 5000 LE Tilburg, The Netherlands}
\footnotetext{Correspondence to : \texttt{M.Seminaroti@cwi.nl} (M.~Seminaroti), \texttt{M.Laurent@cwi.nl} (M.~Laurent), CWI, Postbus 94079, 1090 GB, Amsterdam. Tel.:+31 (0)20 592 4386.}


\maketitle

\begin{abstract}
We present a new efficient combinatorial algorithm for recognizing if a given symmetric matrix is Robinsonian, i.e., if its rows and columns can be simultaneously reordered so that entries are monotone nondecreasing in rows and columns when moving toward the diagonal.
As main ingredient we introduce a new algorithm, named Similarity-First-Search (SFS), which extends  Lexicographic Breadth-First Search (Lex-BFS) to weighted graphs and which we use in a multisweep algorithm to recognize Robinsonian matrices.
Since Robinsonian binary matrices correspond to unit interval graphs, our algorithm can be seen as a generalization to weighted graphs
of the  3-sweep Lex-BFS algorithm  of Corneil  for recognizing unit interval graphs.
This new recognition algorithm is extremely simple and it exploits new insight on the  combinatorial structure of Robinsonian matrices.
For an $n\times n$  nonnegative matrix with $m$ nonzero entries, it terminates in $n-1$ SFS sweeps,  with overall running time $O(n^2 +nm\log n)$.

\medskip

\noindent
\textbf{Keywords:}
\textit{Robinson (dis)similarity; partition refinement; seriation; Lex-BFS; LBFS; Similarity Search}
\end{abstract}

\graphicspath{{./figures/}}

\section{Introduction} 

The {seriation problem}, introduced by Robinson \cite{Robinson51} for chronological dating, is a classic and well known sequencing problem, where the goal is to order a given set of objects in such a way that similar objects are ordered close to each other (see e.g.~\cite{Innar10} and references therein for details). 
This problem arises in many applications where objects are given through some information about their pairwise similarities (or dissimilarities) (like in data about user ratings, images, sounds, etc.).

The seriation problem can be formalized using a special class of matrices, namely Robinson matrices.
A symmetric matrix $A=(A_{xy})_{x,y=1}^n$ is  a {\em Robinson similarity matrix}  if its entries are monotone nondecreasing  in the rows and columns when moving toward the main diagonal, i.e., if 
$A_{xz}\le \min\{A_{xy},A_{yz}\}$  for all $1\le x < y < z \le n$.
Given a set of $n$ objects to order and a symmetric matrix $A=(A_{xy})$ whose entries  represent their pairwise similarities,
the seriation problem asks to find a permutation~$\pi$ of $[n]$ so that the matrix $A_{\pi}=(A_{\pi(x)\pi(y)})$, obtained by permuting  both the rows and columns of $A$  simultaneously according to $\pi$, is a Robinson matrix.
The matrix~$A$ is said to be a {\em Robinsonian similarity matrix} if such a permutation exists.

\medskip
The Robinsonian structure is a strong property and, even though it might be desired in some problems, the data could be affected by noise, leading to the need to solve seriation in presence of error.
Finding a Robinsonian matrix which is closest in the $\ell_\infty$-norm to a given similarity matrix is an NP-hard problem~\cite{Chepoi09}. We refer to~\cite{Chepoi11} for an approximation algorithm  and  to~\cite{Goulermas15,Hahsler08}  for approaches to this problem. 
Nevertheless, Robinsonian matrices play an important role also when data is affected by noise, as Robinsonian recognition algorithms can be used as core subroutines to design efficient heuristics or approximation algorithms for solving seriation in presence of errors (see, e.g., \cite{Chepoi11,Fogel15}). 
In this paper we consider the problem of recognizing whether a given $n\times n$ matrix is Robinsonian.

In the past years, different recognition algorithms for Robinsonian matrices have been studied.  
The first polynomial algorithm to recognize Robinsonian matrices was introduced by Mirkin and Rodin \cite{Mirkin84}.
It is based on the characterization of Robinsonian matrices in terms of interval hypergraphs, and it uses the PQ-tree algorithm of Booth and Leuker~\cite{Booth76} as core subroutine,  with an overall running time of $O(n^4)$. 
Chepoi and Fichet \cite{Chepoi97} introduced later a simpler algorithm 
using a divide-and-conquer strategy applied to preprocessed data obtained by sorting the entries of $A$, lowering the running time to $O(n^3)$.
Using the same sorting preprocessing, Seston \cite{Seston08} improved the complexity of the recognition algorithm to $O(n^2\log n)$.
Recently, Pr\'ea and Fortin \cite{Prea14} presented an optimal $O(n^2)$  algorithm, using the algorithm from Booth and Leuker~\cite{Booth76} to compute a first PQ-tree which they update throughout the algorithm. While all these algorithms use the connection to interval graphs or hypergraphs, 
in our previous work \cite{Laurent15} we presented a  recursive recognition algorithm exploiting 
a connection to unit interval graphs and with core subroutine 
Lexicographic Breadth-First Search (Lex-BFS or LBFS), a special version of Breadth-First Search (BFS) introduced by Rose and Tarjan \cite{Rose75}.
The algorithm of \cite{Laurent15}  is suitable for sparse matrices and it runs in $O(d(m+n))$ time, where $m$ is the number of nonzero entries of $A$ and $d$ is the depth of the recursion tree computed by the algorithm, which is upper bounded by the number of distinct nonzero entries  of $A$.

While all the above mentioned recognition algorithms are combinatorial, Atkins et al. \cite{Atkins98} presented earlier a numerical spectral algorithm, based on reordering the entries of the second smallest eigenvector of the Laplacian matrix associated to $A$ (aka the Fiedler vector).
Given its simplicity, this algorithm is used in some classification applications (see, e.g., \cite{Fogel14}) as well as in spectral clustering (see, e.g., \cite{Barnard93}), and it runs in $O(n(T(n)+n\log n))$ time, where $T(n)$ is the complexity of computing (approximately) the eigenvalues of an $n\times n$ symmetric matrix. 

Note that  the algorithms in \cite{Atkins98}, \cite{Prea14} and \cite{Laurent15} also return all the possible Robinson orderings of a given Robinsonian matrix $A$, which can be useful in some practical applications.

\medskip
In this paper we introduce a new combinatorial recognition algorithm for Robinsonian matrices.
As a main ingredient, we define a  new exploration algorithm for weighted graphs, named \textit{Similarity-First Search} (SFS), which is a generalization of the classical Lex-BFS algorithm to weighted graphs.
Intuitively, the SFS algorithm explores vertices of a weighted graph in such a way that most similar vertices (i.e., corresponding to largest edge weights) are visited first,
while still respecting the priorities imposed by previously visited vertices.
When applied to an unweighted graph (or equivalently to a binary matrix),  the SFS algorithm reduces to Lex-BFS.
As for Lex-BFS, the SFS algorithm is entirely based on a unique simple task, namely partition refinement,  a basic operation about sets which can be  implemented efficiently (see \cite{Habib00} for  details).

We will use  the SFS algorithm to define our new  Robinsonian recognition algorithm. Specifically, we introduce a multisweep algorithm, where each sweep uses the order returned by the previous sweep to break ties in the (weighted) graph search.
Our main result in this paper is  that our multisweep algorithm can recognize after at most  $n-1$ sweeps whether a given $n\times n$ matrix $A$ is Robinsonian. Namely we will show that the last sweep is a Robinson  ordering of $A$ if and only if the matrix $A$ is Robinsonian. 
Assuming that the matrix $A$ is nonnegative and given as an adjacency list of an undirected weighted graph with $m$ nonzero entries, our algorithm runs in $O( n^2+ m n \log n)$ time.

Multisweep algorithms are well studied approaches to recognize classes of (unweighted) graphs (see, e.g.,~\cite{Corneil16}).
In the literature there exist many results on multisweep algorithms based on Lex-BFS and its variants.
For example, cographs can be recognized in 2 sweeps~\cite{Bretscher08}, unit interval graphs can be recognized in 3 sweeps~\cite{Corneil04} and interval graphs can be recognized in at most 5 sweeps~\cite{Corneil09}.
Very recently, Dusart and Habib \cite{Dusart15} introduced  a multisweep algorithm to recognize in at most $n$ sweeps cocomparability graphs.
For a more exhaustive list of multisweep algorithms please refer to \cite{Corneil05,Corneil09}.

As a graph is a unit interval graph if and only if its adjacency matrix is Robinsonian  \cite{Roberts69},
the 3-sweep recognition algorithm for unit interval graphs  of Corneil \cite{Corneil04} is in fact our main inspiration and motivation to develop a generalization of Lex-BFS for weighted graphs.

To the best of our knowledge, 
the present paper is the first work introducing and studying explicitely the properties of a multisweep search algorithm for weighted graphs.
The only related idea that we could find is about replacing BFS with Dijkstra's algorithm, which is  only briefly mentioned in \cite{Habib13}.

The relevance of this work is twofold.
First, we reduce the Robinsonian recognition problem to  an extremely simple and basic operation, namely to partition refinement.
Hence, even though from a theoretical point of view the algorithm is computationally slower than the optimal one presented in \cite{Prea14},
its simplicity makes it easy to implement and thus hopefully will encourage the use and the study of Robinsonian matrices in more practical problems.
Second, we introduce a  new (weighted) graph search, which we believe is of independent interest and could potentially be used for the recognition of other structured matrices or just as basic operation in the broad field of `Similarity Search'. 
In addition, we introduce some new concepts extending analogous notions in graphs, like the notion of  `path avoiding a vertex' and `anchors' of Robinson orderings, which capture well the combinatorial structure of Robinsonian matrices.
As an example,  we give combinatorial characterizations for the end points (aka anchors) of Robinson orderings.

\subsection*{Contents of the paper} 

The paper is organized as follows.
Section~\ref{sec:2-preliminaries} contains some preliminaries. 
In Subsection~\ref{sec:2-basic facts} we give  basic facts about Robinsonian matrices and Robinson orderings and we 
introduce several  concepts (path avoiding a vertex,
valid vertex, anchor) 
playing a crucial role in the paper.   
 Subsection~\ref{sec:2-anchor characterization} contains   combinatorial characterizations for (opposite) anchors of Robinsonian matrices.

Section \ref{sec:3-SFS algorithm} is devoted to  the SFS algorithm. First, we describe the algorithm in Subsection~\ref{sec:3-description algorithm}  and   we  characterize  SFS orderings  in Subsection \ref{sec:3-characterization SFS}.
Then, in Subsection \ref{sec:3-PAL} we introduce a fundamental lemma which we will use throughout the paper, named the `Path Avoiding Lemma'.
Finally,  in Subsection \ref{sec:3-end points}  we introduce the notion of `good SFS ordering' 
and we show properties of end-vertices of (good) SFS orderings, namely that they are (opposite) anchors of Robinsonian matrices.

In Section \ref{sec:4-SFS+ algorithm} we discuss the  variant $\SFS_+$ of the SFS algorithm,  an extension of Lex-BFS$_+$ to weighted graphs, which differs from SFS in the way ties are broken
The $\SFS_+$ algorithm takes a given ordering as input which it uses to break ties.
In Subsection \ref{sec:4-good SFS and end points} we show a basic property of the $\SFS_+$ algorithm, namely 
that it `flips' the end points of the input ordering.
Then in Subsection \ref{sec:4-similarity layers}  we introduce the  `similarity layers' of a matrix, a strengthened version of BFS layers for unweighted graphs, which are  useful for the correctness proof of the multisweep algorithm. We show in particular that the similarity layers enjoy some compatibility with Robinson  and $\SFS_+$ orderings.

In Section \ref{sec:5-the multisweep algorithm} we present the multisweep algorithm to recognize Robinsonian matrices and we prove its correctness. 
In  Subsection \ref{sec:5-description} we describe  the multisweep algorithm and show that it terminates in 3 sweeps when applied to a binary matrix, thus giving a new proof of the result of Corneil \cite{Corneil04} for unit interval graphs.
In Subsection \ref{sec:5-three-good SFS ordering} we study properties of  `3-good SFS orderings', which are orderings obtained after three $\SFS_+$ sweeps.  In particular we show that they contain classes of Robinson triples and that, after deleting their end points, they induce good SFS orderings, which will enable us to apply induction in the correctness proof.
After that we have all the ingredients needed to  conclude the correctness proof for the multisweep algorithm, 
we show  in Subsection  \ref{sec:5-final proof} that it can recognize  in at most $n-1$ sweeps whether an $n\times n$ matrix is Robinsonian. 
Furthermore, we present in Subsection~\ref{sec:5-worst case instances} a family of $n \times n$ Robinsonian matrices (communicated to us by S. Tanigawa) for which the SFS multisweep algorithm  requires exactly $n-1$ sweeps.

Finally, in Section \ref{sec:6-complexity} we discuss the complexity of the SFS algorithm, and we conclude with remarks and open questions  in  Section \ref{sec:7-conclusions}.

\section{Preliminaries}\label{sec:2-preliminaries}

In this section we  introduce some notation and recall some basic properties and definitions for unit interval graphs and Robinsonian matrices.
In particular, we introduce the concepts of `path avoiding a vertex' and  `valid vertex'  and  we give  combinatorial characterizations for  end points of Robinson orderings (also named `anchors')  and for `opposite anchors',  which will play an important role in the rest of the paper.

\subsection{Basic facts}\label{sec:2-basic facts}

Let  $\pi$ be a linear order of $V=[n]$. For two distinct elements $x,y \in [n]$, the notation  $x <_{\pi} y$ means that $x$ appears before $y$ in $\pi$
and, for  disjoint subsets $U,W\subseteq V$,   $U <_{\pi} W$ means that  $x <_{\pi} y$ for all $x\in U,$ $y\in W$.
The linear order~$\pi$  is a permutation of $[n]$, which can be represented  as a sequence $(x_1,\ldots,x_n)$ with  $x_1<_\pi \ldots <_\pi x_n$, and $\pi^{-1}$ is the  reverse linear order 
$(x_n,x_{n-1},\ldots, x_1)$. 
An ordered partition $\phi=(B_1,\dots,B_r)$ of a ground set $V$ is an ordered collection of disjoint subsets of $V$ whose union is $V$.

Throughout,  $\MS^n$ denotes the set of symmetric $n\times n$ matrices. 
Given $A \in \MS^n$ and a subset $S\subseteq [n]$,  $A[S]=(A_{xy})_{x,y\in S}$ is the principal submatrix of~$A$ indexed by~$S$.
A symmetric matrix $A \in \MS^n$ is called  a {\em Robinson similarity matrix}  if its entries are monotone nondecreasing
 in the rows and columns when moving towards the main diagonal, i.e., if 
\begin{equation}\label{eq:Robinson inequalities}
A_{xz}\le \min\{A_{xy},A_{yz}\} \quad \text{for all} \quad 1 \le  x < y < z \le n.
\end{equation}
Note that the diagonal entries of $A$ do not play a role in the above definition. 
If there exists a permutation $\pi$ of $[n]$ such that the matrix $A_{\pi}:=(A_{\pi(x)\pi(y)})_{x,y=1}^n$, obtained by permuting  both the rows and columns of $A$  simultaneously according to $\pi$, is a Robinson matrix then $A$ is said to be a {\em Robinsonian similarity} and 
 $\pi$ is  called a {\em Robinson ordering} of $A$.
In the literature, a distinction is made between Robinson(ian) similarities and Robinson(ian) dissimilarities. 
A symmetric matrix $A$ is called a {\em Robinson dissimilarity matrix} if its entries are monotone nondecreasing 
in the rows and columns when moving away from the main diagonal.
Hence  $A \in \MS^n$ is a Robinson(ian) similarity precisely when  $-A$ is a Robinson(ian) dissimilarity and thus the properties extend directly from one class to the other one.
For this reason, in this paper we will deal exclusively with Robinson(ian) similarities.
Hence, when speaking of a Robinson(ian) matrix, we mean a Robinson(ian) similarity matrix.
Furthermore, with  $J\in \mathcal S^n$ denoting the all-ones matrix, it is clear that 
 if $A$ is a Robinson(ian) matrix then  $A+ \lambda J$ is also a Robinson(ian) matrix for any scalar $\lambda$.
Hence, we may consider, without loss of generality,  nonnegative similarities $A$ (whose smallest  entry is equal to 0).

\medskip

In order to fully understand Robinsonian matrices and the motivation for our work, it is useful to briefly discuss the special class of binary Robinsonian matrices. 
Any  symmetric matrix $A\in \{0,1\}^{n\times n}$ corresponds to a graph $G=(V=[n],E)$ whose edges are the positions of the nonzero entries of $A$. 
Then it  is well known that $A$ is a Robinsonian similarity if and only if $G$ is a  unit interval graph \cite{Roberts69}.
A graph $G=(V=[n],E)$ is called a {\em unit interval graph} if its vertices can be mapped to unit intervals $I_1,\ldots,I_n$ of the real line such that two distinct vertices $x,y\in V$ are adjacent in $G$ if and only if $I_x\cap I_y\ne \emptyset$.
There exist several equivalent characterizations for unit interval graphs. The following one highlights  the analogy between unit interval graphs and Robinson orderings.

\begin{theorem}[\textbf{3-vertex condition}]\cite{Looges93}
A graph $G=(V,E)$ is a unit interval graph if and only if there exists a linear ordering $\pi$ of $V$ such that, for all $x,y,z\in V$,
\begin{equation}\label{eq:3-vertex condition}
 x<_\pi y<_\pi z, \ \{x,z\}  \in E \Longrightarrow \{x,y\},\{y,z\}  \in E.
 \end{equation}
\end{theorem}

It is clear that, for a binary matrix $A\in \mathcal S^n$, condition~(\ref{eq:Robinson inequalities}) is equivalent to~(\ref{eq:3-vertex condition}).
This equivalence and the fact that unit interval graphs can be recognized with a Lex-BFS multisweep algorithm \cite{Corneil04} motivated us  to find an extension of Lex-BFS to weighted graphs and to use it  to obtain a (simple) multisweep  recognition algorithm for Robinsonian matrices.

\medskip
Given the analogy with unit interval graphs, it will be convenient to view symmetric matrices as weighted graphs. Namely,  
any nonnegative symmetric matrix $A \in \MS^n$ corresponds to the weighted graph  $G=(V=[n],E)$   whose  edges are the pairs $\{x,y\}$ with $A_{xy}> 0$, 
with edge weights $A_{xy}$. 
Again, the assumption of nonnegativity can be made without loss of generality and is for convenience only.
Accordingly  we will often refer to the elements of $V=[n]$ indexing $A$ as  vertices (or nodes). 
For $x\in V$, $N(x)=\{y\in V\setminus \{x\}: A_{xy}> 0\}$ denotes the neighborhood of $x$ in $G$.

In what follows we will extend some graph concepts to the general setting of weighted graphs (Robinsonian matrices). 
Throughout the paper, we will  point out   links between our results and some  corresponding known results  for Lex-BFS applied to graphs and we will mostly  refer to \cite{Corneil09} where more complete  references about Lex-BFS can be found.

\medskip
We now introduce some notions  and simple facts about Robinsonian matrices and orderings. 
Consider a matrix $A\in \mathcal S^n$. Given distinct elements $x,y,z\in V$, the  triple $(x,y,z)$ is said to be  {\em Robinson} if it  satisfies  (\ref{eq:Robinson inequalities}), i.e.,  if
$A_{xz}\le \min\{A_{xy},A_{yz}\}$. Given a set $S\subseteq V$ and $x\in V\setminus S$, we say that  $x$ is \textit{homogeneous} with respect to~$S$ if $A_{xy} = A_{xz}$ for all  $y,z \in S$ (extending the corresponding notion for graphs, see, e.g.,~\cite{Corneil09}).
 The following is  an  easy necessary condition for the Robinson property. 

\begin{lemma}\label{thm:non-Robinson property}
Let $A \in \MS^n$ be a Robinsonian similarity. Assume that  there exists a Robinson ordering $\pi$ such that $x <_{\pi} z <_{\pi} y$. Then $A_{uz} \geq \min\{A_{ux},A_{uy}\}$ for all $u \neq x,y,z \in [n]$.
\end{lemma}
\begin{proof}
Indeed, $u<_\pi z$ implies  $u<_\pi z<_\pi y$ and thus $A_{uz}\ge A_{uy}$,  and $z<_\pi u$ implies $x<_\pi z<_\pi u$ and thus $A_{uz}\ge A_{ux}$.
$\qquad$
\end{proof}

We now make  a simple observation on how three elements $x,y,z\in V$ may  appear in a Robinson ordering $\pi$ of $A$ depending on their similarities.
Namely, if we have $A_{xz} > \min \{A_{xy},A_{yz}\}$ then, either $y$ comes before both $x$ and $z$ in $\pi$, or $y$ comes after both $x$ and $z$ in $\pi$.
In other words, if $x$ and $z$ are more similar to each other than to~$y$, then $y$ cannot be ordered between $x$ and $z$ in any Robinson ordering $\pi$.
Moreover, if $A_{xz} < \min \{A_{xy},A_{yz}\}$ then, either $x <_{\pi} y<_{\pi} z$, or $z <_{\pi} y<_{\pi} x$.
In other words, if~$x$ and~$z$ are more similar to $y$ than to each other, then $y$ must be ordered between $x$ and $z$ in any Robinson ordering $\pi$.

This observation motivates the following  notion of  `path avoiding a vertex', which will play a central role in our discussion.
Note that this notion is closely  related to the notion of `path missing a vertex'  for Lex-BFS \cite{Corneil09}, although it is not equivalent to it when applied to a  binary matrix.
Note also that in our setting the notion of  path is defined for a matrix and a path is just a sequence of (possibly repeated)  vertices.

\begin{definition}[\textbf{Path avoiding a vertex}]\label{defpath}
Given distinct elements $x,y,z \in V$, a {\em path from $x$ to $z$ avoiding $y$} is a sequence $P=(x=v_0, v_1,\ldots, v_{k-1},v_k=z)$ of (not necessarily distinct) elements of $V$ where each triple $(v_i,y,v_{i+1})$ is not Robinson, i.e., 
\begin{equation*}
A_{v_iv_{i+1}}>\min\{A_{yv_i}, A_{yv_{i+1}}\},  \quad \forall \ i =0,1,\ldots, k-1.
\end{equation*}
We let $|P|=k+1$ denote the length of the path $P$ (i.e., its number of elements).
\end{definition}

\begin{figure}[!h]
\centering
\includegraphics{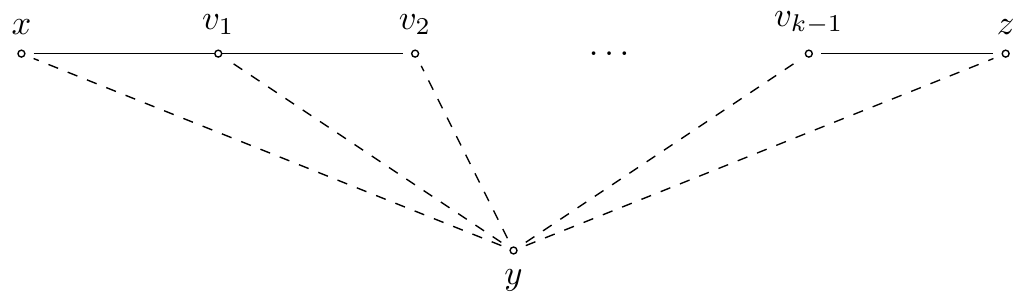}
\label{fig:path}
\caption{A path from $x$ to $z$ avoiding $y$: each continuous line indicates a value which is strictly larger than the minimum of the two adjacent dotted lines}
\end{figure}

The following simple but useful property holds.

\begin{lemma}\label{thm:path avoiding and Robinson ordering}
Let $A \in \MS^n$ be a Robinsonian matrix. If there exists a path from $x$ to $z$ avoiding $y$, then $y$ cannot lie between $x$ and $z$ in any Robinson ordering $\pi$ of $A$.
\end{lemma}
\begin{proof}
Let $(x=v_0, v_1,\ldots, v_{k-1},v_k=z)$ be a path from~$x$ to~$z$ avoiding~$y$.
Then, by definition, we have $A_{v_iv_{i+1}}>\min\{A_{yv_i}, A_{yv_{i+1}}\} \ \text{ for all } i =0,1,\ldots, k-1$, and thus~$y$ cannot appear between $v_i$ and $v_{i+1}$ in any Robinson ordering $\pi$.
Hence $y$ cannot lie between $x$ and $z$ in any Robinson ordering $\pi$.
$\qquad$
\end{proof}

We now introduce the  notion of `valid vertex' which we will use in throughout the section to characterize end points of Robinson orderings.

\begin{definition}[\textbf{Valid vertex}]\label{defvalid}
Given a matrix $A\in \mathcal S^n$, an element $z\in V$ is said to be {\em valid} if, for any distinct elements $u,v\in V\setminus \{z\}$, there do not exist both a path from $u$ to $z$ avoiding $v$ and a path from $v$ to $z$ avoiding $u$.
\end{definition}

Observe  that, if $z \in V$ is a valid vertex of a 
matrix $A$ and $S \subseteq V$ is a subset containing $z$, then $z$ is also a valid vertex of $A[S]$.
It is easy to see that, for a  $0/1$ matrix, the above definition of valid vertex coincides with the notion of valid vertex for Lex-BFS \cite{Corneil09}.

Consider, for example, the following matrix (already ordered in a Robinson form):
\begin{equation*}
A=
\bordermatrix{
~ & \textbf{a} & \textbf{b} & \textbf{c} & \textbf{d} & \textbf{e} & \textbf{f} & \textbf{g} \cr
\textbf{a}  & * & 7 & 6 & 0 & 0 & 0 & 0 \cr
\textbf{b}  & & * & 7 & 3 & 2 & 1 & 1 \cr
\textbf{c}  & & & * & 7 & 2 & 2 & 1 \cr
\textbf{d}  & & & & * & 3 & 3 & 3 \cr
\textbf{e}  & & & & & * & 7 & 5 \cr
\textbf{f}  & & & & & & * & 6 \cr
\textbf{g}  & & & & & & & * \cr
}
\end{equation*}
Then the vertex $d$ is not valid. Indeed,  for the  two vertices $a$ and $g$,  there exist a path from $a$ to $d$ avoiding $g$ and a path from $g$ to $d$ avoiding $a$; namely the path $(d,b,a)$ avoids $g$ and the path $(d,b,g)$ avoids $a$ (see Figure~\ref{fig:vertex not admissible}). 

\begin{figure}[!h]
\centering
\includegraphics[page=1]{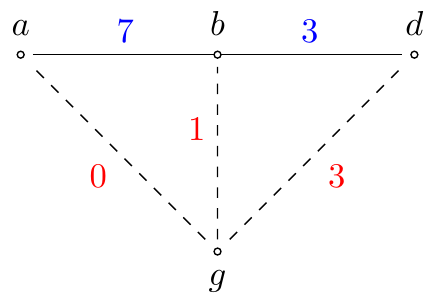}
\qquad
\includegraphics[page=2]{vertex_not_admissible.pdf}
\caption{Element $d$ is not valid}
\label{fig:vertex not admissible}
\end{figure}

\subsection{Characterization of anchors}\label{sec:2-anchor characterization}

In this subsection we introduce the notion of `(opposite) anchors'  of a Robinsonian matrix and then we give characterizations in terms of valid vertices.
The notion of anchor was used  for unit interval graphs  in \cite{Corneil95} (where it refers to an end point of a linear order satisfying the 3-vertex condition~(\ref{eq:3-vertex condition})) and it should not be confused with the notion of end-vertex used for interval graphs in \cite{Corneil09}  (where it  refers to an end point of a Lex-BFS ordering, see~\cite{Corneil10} for more details).

\begin{definition}[\textbf{Anchor}]\label{defanchor}
Given a Robinsonian similarity $A \in \MS^n$, a vertex $a \in [n]$ is called an \emph{anchor} of $A$ if  there exists a Robinson ordering $\pi$ of $A$ whose last vertex is  $a$.
Moreover, two distinct vertices $a,b$ are called {\em opposite anchors} of $A$ if there exists a Robinson ordering $\pi$ of $A$ with $a$ as first vertex and $b$ as last vertex.
\end{definition}

Hence, an anchor is an end point of a Robinson ordering.
Clearly, every Robinsonian matrix has at least one pair of opposite  anchors. 
It is not difficult to see that every anchor must be valid.
We now show that conversely every valid vertex is an anchor.
This is the analogue of  \cite[Lemma 2]{Corneil05} for Lex-BFS over interval graphs.

\begin{theorem}\label{thm:anchor admissible}
Let $A \in \MS^n$ be a Robinsonian matrix. Then a vertex $z \in V$ is an anchor of $A$ if and only if it is valid. 
\end{theorem}

\begin{proof}
($\Rightarrow$)
Assume $z$ is an anchor of $A$ and let $\pi$ be a Robinson ordering of $A$ with $z$ as last element.
Suppose for contradiction that, for some elements $u,v\in V$,  there exist both a path $P$ from $u$ to $z$ avoiding $v$ and a path $Q$ from $v$ to $z$ avoiding~$u$.
Using Lemma \ref{thm:path avoiding and Robinson ordering}  and the path $P$, we obtain that  that  $v$ lies before $u$ or after $z$ in $\pi$, and using the path $Q$ we obtain that  $u$ lies before $v$ or after $z$ in~$\pi$. 
As $z$ is the last element of $\pi$, we must have $v<_\pi u$ in the first case and $u<_\pi v$ in the second case, which is impossible. 

($\Leftarrow$)
Conversely, assume that  $z$ is valid; we show that $z$ is an anchor of $A$. The proof is by induction on the size $n$ of the matrix $A$. 
The result holds clearly when $n=2$. So we now assume $n\ge 3$ and that the result holds for any Robinsonian matrix of order at most $n-1$.
We need to construct a Robinson ordering $\pi'$ of $A$ with $z$ as last vertex.
For this we consider a Robinson ordering $\pi$ of $A$. 
We let $x$ denote its first element and $y$ denote its last element. 
If $z=x$ or $z=y$, then we would be done.
Hence we may assume $x<_\pi z <_\pi y$.
For any $v <_{\pi} z$, we denote by $P_{\pi}(v,z)$ the path from $v$ to $z$ consisting of the sequence of vertices appearing consecutively between $v$ and $z$ in $\pi$.

We now define the following two sets:
\begin{align}\label{relBC}
 \MB=\{v <_{\pi} z: P_{\pi}(v,z) \text{ avoids } y\}, \ \
 \MC= \{v <_{\pi} z: v \notin \MB\}. 
\end{align}
Next we show their following properties, which will be useful to conclude the proof.
\begin{claim}\label{claim1}
The following holds:
\begin{itemize}
\item[(i)]
For any $v\in \MB$, $A_{vy}=A_{yz}.$
\item[(ii)]
If  $v\in \MB$ and $v <_\pi u <_{\pi} z$, then $u\in \MB$.
\item[(iii)]
Any element  $v \in \MC$ is homogeneous with respect to $V\setminus \MC$, i.e., $A_{vw}=A_{vw'}$ for all $w,w'\in V\setminus \MC$.
\end{itemize}
\end{claim}
\begin{proof}
(i) 
As $v<_\pi z<_\pi y$, then $A_{vy}\le A_{yz}$. 
We show that equality holds. Suppose not, i.e.,  $A_{vy} < A_{yz}$. Then $Q=(y,z)$ is a path from $y$ to $z$ avoiding $v$.
Since $v \in \MB$,  $P=P_{\pi}(v,z)$ is a path from $v$ to $z$ avoiding $y$, and thus the existence of the paths $P,Q$ contradicts the assumption that $z$ is valid. Hence we must have $A_{vy}=A_{yz}$.

(ii)
If $v \in \MB$ then $P_{\pi}(v,z)$ avoids $y$ and thus the subpath $P_{\pi}(u,z)$ also avoids $y$, which implies $u \in \MB$.

(iii) Let $u\in \MB$ denote the element of $\MB$ appearing first in the Robinson ordering~$\pi$.
Then, by (ii), for any $v\in \MC$, $v <_\pi u<_\pi y$ and thus $A_{vy}\le A_{vu}$ by definition of Robinson ordering.
Hence, in order to show that $v$ is homogeneous with respect to $V\setminus \MC$, it suffices to show that $A_{vu}=A_{vy}$ (as, using the Robinson ordering property, this would in turn imply that $A_{vw}=A_{vw'}$ for all $w,w'\in V\setminus \MC$).
Suppose for contradiction that 
there exists $v\in\MC$ such that $A_{vu}\ne A_{vy}$, and let 
$v$ denote the element of $\MC$ appearing last in $\pi$ with $A_{vu}\ne A_{vy}$.

Then $A_{vu}>A_{vy}$ and the path $(v,u)$ avoids $y$. 
Since $P_\pi(u,z)$ is a path from~$u$ to $z$ avoiding $y$ (because $u\in \MB$),  then the path  $P=\{v\} \cup P_{\pi}(u,z)$ (obtained by concatenating $(v,u)$ and $P_\pi(u,z)$) is a path from $v$ to $z$ avoiding $y$. 
This implies that~$v$ and $u$ cannot be consecutive in $\pi$, as otherwise we would have $v \in \MB$, contradicting the fact that $v\in \MC$.
Hence, there exists $v'\in \MC$ such that $v <_{\pi} v' <_{\pi} u$. 
By the maximality assumption on $v$, it follows that $A_{v'u}=A_{v'y}$.

As  $z$ is valid and  $P=\{v\} \cup P_{\pi}(u,z)$ is a path from $v$ to $z$ avoiding $y$, 
it follows that no path from $y$ to $z$ can avoid $v$. 
In particular, the path $(y,z)$ does not avoid~$v$ and thus it must be $A_{yz} \leq \min \{A_{vy},A_{vz}\}$.
Recall that we assumed $A_{vu}>A_{vy}$. 
As $v<_\pi v' <_\pi u <_\pi z <_\pi y$, combining the above inequalities with the inequalities coming from the Robinson ordering $\pi$, we obtain
$A_{v'y}\le A_{yz}\le A_{vy}<A_{vu}\le A_{v'u}$, 
 which contradicts the equality $A_{v'u}=A_{v'y}$.
$\qquad$
\end{proof}

We now turn to the set of vertices coming after $z$ in $\pi$. Symmetrically with respect to $z$, we can define  the analogues of the sets $\MC, \MB$ defined in (\ref{relBC}), which we denote by~$\MC', \MB'$.
For this replace $\pi$ by its reverse ordering $\pi^{-1}$ and $y$ by $x$ (the first element of $\pi$ and thus the last element of $\pi^{-1}$), i.e., set
\begin{align*}
 \MB'=\{v >_{\pi} z: P_{\pi}(z,v) \text{ avoids } x\}, \ \ 
 \MC'= \{v >_{\pi} z: v \notin \MB'\}.
\end{align*}

To recap, we have that $\pi = (\MC,\MB,z,\MB',\MC')$.
Recall that $x$ and $y$ are respectively the first and the last vertex in $\pi$.
Note that it cannot be that $\MC =\MC'=\emptyset$, as this would imply that $x \in \MB$ and $y \in \MB'$, and thus this would contradict the fact that $z$ is valid (using the definition of the two sets $\MB$ and $\MB'$).
Therefore, we may assume (without loss of generality) that  $\MC \neq \emptyset$. 
Let $v$ be the vertex of $\MC$ appearing last in the Robinson ordering $\pi$. 
By Claim \ref{claim1} (iii),
 $v$ is homogeneous with respect to the set $S=V\setminus \MC$, i.e., 
all entries $A_{vw}$ take the same value for any $w\in S$.

Consider the matrix $A[S]$, the principal submatrix of $A$ with rows and columns in $S$. As $|S|\le n-1$ and $z$ is valid (also with respect to $A[S]$), we can conclude using the induction assumption that $z$ is an anchor of $A[S]$.   Hence, there exists  a Robinson ordering $\sigma$ of $A[S]$ admitting  $z$ as last element.
  
Now, consider the linear order $\pi'=(\pi[\MC],\sigma)$ of $V$ obtained by concatenating first the order $\pi$ restricted to $\MC = V \setminus S$ and second the linear order $\sigma$ of $S$. Using the fact that every vertex in $\MC$ is homogeneous to all elements of $S$, 
we can conclude that the new linear order $\pi'$ is a Robinson ordering of the matrix $A$. As $z$ is the last element of $\pi'$, this shows that $z$ is an anchor of $A$ and thus concludes the proof. 
$\qquad$
\end{proof}

The above proof can  be extended to characterize pairs of opposite anchors.

\begin{theorem}\label{thm:characterization opposite anchors}
Let $A \in \MS^n$ be a Robinsonian matrix. Two distinct vertices $z_1, z_2 \in [n]$ are opposite anchors of  $A$ if and only if they are both valid and there does not exist a path from $z_1$ to $z_2$ avoiding any other vertex. 
\end{theorem}

\begin{proof}
($\Rightarrow$)
Assume that  $z_1$ and $z_2$ are opposite anchors. Then they are both anchors and thus, in view of Theorem \ref{thm:anchor admissible}, they are both valid. Let $\pi$ a Robinson ordering starting with $z_1$ and ending with $z_2$. Suppose, for the sake of contradiction, that there exists a vertex $x$ and a path from $z_1$ to $z_2$ avoiding $x$.
Then, by Lemma~\ref{thm:path avoiding and Robinson ordering}, $x$ cannot lie in $\pi$ between $z_1$ and $z_2$, yielding  a contradiction.

($\Leftarrow$)
Assume that $z_1$ and $z_2$ are valid and that there does not exist a path from $z_1$ to $z_2$ avoiding any other vertex.
We show that they are opposite anchors. 
Consider a Robinson ordering $\pi$ of $A$ whose first element is $z_1$ and call $y$ its last element.
If $y=z_2$ then we are done. Hence, we may assume that $z_1 <_\pi z_2 <_\pi y$.
As in the proof of Theorem \ref{thm:anchor admissible}, for any $v <_{\pi} z_2$, we denote by $P_{\pi}(v,z_2)$ the path from $v$ to $z_2$ consisting of the sequence of vertices appearing consecutively between $v$ and $z_2$ in~$\pi$.
Then, we can define the sets as in (\ref{relBC}) in the proof of Theorem \ref{thm:anchor admissible}, where $z$ is replaced by $z_2$, i.e.,:
\begin{align*}
 \MB=\{v <_{\pi} z_2: P_{\pi}(v,z_2) \text{ avoids } y\}, \ \
 \MC= \{v <_{\pi} z_2: v \notin \MB\}. 
\end{align*}
By assumption, $z_1\not\in \MB$, else $P_\pi(z_1,z_2)$ would avoid $y$, contradicting the nonexistence of a path from $z_1$ to $z_2$ avoiding any other vertex.
Therefore $z_1\in\MC$ and thus $\MC \neq \emptyset$. Let $S=V \setminus \MC$.
Using the same reasoning as in the proof of Theorem \ref{thm:anchor admissible}, we can now conclude that one can find a Robinson ordering $\sigma$ of $A[S]$, where $S$ contains all the elements coming after the last element of $\MC$ in $\pi$. 
The new linear order $\pi'=(\pi[\MC],\sigma)$ of $V$ obtained by concatenating first the order $\pi$ restricted to $\MC = V \setminus S$ and second the linear order $\sigma$ of $S$ is then a Robinson ordering of $A$ whose first element is $z_1$ and whose last element is $z_2$, which concludes the proof.
$\qquad$
\end{proof}

\section{The SFS algorithm} \label{sec:3-SFS algorithm}

In this section we introduce our new Similarity-First Search (SFS) algorithm. This algorithm will be  applied to a (nonnegative) matrix $A \in \MS^n$ and return a linear order of $V=[n]$, called a {\em SFS ordering} of $A$.
As mentioned above, one can associate to   $A$  a weighted graph $G=(V=[n],E)$, with edges the pairs $\{x,y\}$ such that  $A_{xy}> 0$ and edge weights $A_{xy}$. 
The SFS algorithm can be thus seen as  a search algorithm for weighted graphs.

We first describe  the algorithm in detail  in Subsection \ref{sec:3-description algorithm} and provide  a 3-point characterization  of  SFS orderings in Subsection 
\ref{sec:3-characterization SFS}.
Then in Subsection \ref{sec:3-PAL} we discuss some  properties of  SFS orderings of  Robinsonian matrices. Specifically, we introduce the fundamental `Path Avoiding Lemma' (Lemma \ref{thm:PAL}) which will be used repeatedly throughout the paper. In particular we use it   in Subsection \ref{sec:3-end points}   to show a  fundamental property of SFS orderings, namely that the last element of the SFS ordering of a Robinsonian matrix $A$ is an anchor of $A$. 

\subsection{Description of the SFS algorithm} \label{sec:3-description algorithm}

The SFS algorithm is a generalization of Lex-BFS for weighted graphs. As we will remark later, when applied to a~$0/1$ matrix, the SFS algorithm coincides with  Lex-BFS.
Roughly speaking, the basic idea is to explore a weighted graph by visiting first vertices which are similar to each other (i.e., corresponding to an edge  with largest weight) but respecting the priorities imposed by previously visited vertices.
The algorithm is based on the implementation of Lex-BFS as a sequence of partition refinement steps as  in \cite{Habib00}.

Partition refinement is a simple technique introduced in \cite{Paige87} to refine a given ordered partition $\phi=(B_1,\dots,B_r)$ of the ground set $V$ 
by a subset $W \subseteq V$.
It produces a new ordered partition of $V$ obtained by splitting each class $B_i$ of $\phi$ in two sets, the intersection $B_i \cap W$ and the difference $B_i \setminus W$. 
If one visualizes an ordered partition as a priority list, the idea behind partition refinement is to modify the classes of the ordered partition while respecting the priorities among the vertices.
 
In our new SFS algorithm, we basically operate a sequence of partition refinements. But instead of splitting into two subsets we will split into several subsets.
Specifically, given two ordered partitions $\phi$ and $\psi$, the output will be a new ordered partition which, roughly speaking, is obtained by splitting each class of $\phi$ into its intersections with the classes of $\psi$. The formal definition is as follows.

\begin{definition}[\textbf{Refine}]\label{def:refine}
Let $\phi=(B_1,\ldots, B_r)$ and  $\psi=(C_1,\ldots,C_s)$ be two ordered partitions of a set  $V$ and a subset $W \subseteq V$, respectively.
{\em Refining $\phi$ by $\psi$} creates the new ordered partition of $V$, denoted by 
\textit{Refine$(\phi,\psi)$}, obtained by replacing in $\phi$ each class $B_i$  by the ordered sequence of classes $(B_i\cap C_1, \ldots, B_i\cap C_s,B_i \setminus (C_1 \cup \dots \cup C_s)=B_i\setminus W)$ and keeping only nonempty classes.
\end{definition}

We will use this partition refinement operation in the case when the partition $\psi$ is obtained by partitioning for decreasing values the elements of the neighborhood $N(p)$ of a given element $p$, according to the following definition.

\begin{definition}[\textbf{Similarity partition}]\label{def:layer partition}
Consider a nonnegative matrix $A\in \mathcal S^n$ and an element $p \in [n]$.
Let $a_1>\ldots> a_s>0$ be the distinct values taken by the entries $A_{px}$ of $A$ for $x \in N(p)=\{y \in [n] : A_{py} >0\}$ and, for $i\in [s]$, set $C_i=\{x \in N(p):  A_{px}=a_i\}$.
Then we define $\psi_p=(C_1,\ldots,C_s)$, which we call the \emph{similarity partition} of $N(p)$ with respect to~$p$.
\end{definition}

We can now describe the SFS algorithm.
The input 
is a nonnegative matrix $A \in \MS^n$ and the output is an ordering $\sigma$ of the set $V=[n]$, that we call a {\em SFS ordering} of~$A$.
As in any general graph search algorithm, the central idea of the SFS algorithm is that, at each iteration, a special vertex (called the \textit{pivot}) is chosen among the subset of unvisited vertices (i.e., the subset of vertices that have not been a pivot in prior iterations).
Such vertices are ordered in a queue which defines the priorities for visiting them. 
Intuitively, the pivot is chosen  as the most similar to the visited vertices, but respecting the visiting priorities imposed by previously visited vertices.

\medskip
\begin{algorithm}[H]\label{alg:SFS}
\caption{\textit{\textit{SFS}}$(A)$}
\SetKwInput {KwIn}{input}
\SetKwInput {KwOut}{output}
\KwIn{a nonnegative matrix $A \in \mathcal{S}^{n}$}
\KwOut{a linear order $\sigma$ of $[n]$}
\vspace{2ex}
$\phi=(V) \leftarrow$ queue of unvisited vertices\\
\For{$i=1,\dots,n$}{
	$S$ is the first class of $\phi$\\
	choose $p$ arbitrarily in $S \leftarrow $ new pivot \label{alg: ties}\\
	$\sigma(p)=i$ $\leftarrow$ let $p$ appear at position $i$ in $\sigma$\\
	remove $p$ from $\phi$\\
	$N(p)$ is the set of vertices $y\in \phi$ with $A_{py}>0$\\
	$\psi_p$ is the similarity partition of $N(p)$ with respect to $p$\\
	$\phi=$\textit{Refine} $(\phi,\psi_p)$
}
\Return: $\sigma$
\end{algorithm}
\medskip

We now discuss in detail how the algorithm works.
In the beginning, all vertices in $V$ are unvisited, i.e., the queue $\phi$ of unvisited vertices is initialized with the unique class $V$.

At the  iteration $i$, we are given an element $p_{i-1}$ (which is the pivot chosen at iteration $i-1$) and a queue $\phi(p_{i-1})=(B_1, \dots, B_r)$, which is an ordered partition of the set of unvisited vertices.
There are two main tasks to perform: the first task is to select the new pivot $p_{i}$, and the second task is to update the queue $\phi(p_{i-1})$ in order to obtain the new queue $\phi(p_{i})$.

The first task is carried out as follows. 
As in the standard {Lex-BFS}, we denote by~$S$ the {\em slice} induced by $p_{i-1}$ (i.e., the last visited vertex), which consists  of the vertices among which to choose the next pivot $p_{i}$. 
The slice $S$ coincides exactly with the first class $B_1$ of $\phi(p_{i-1})$.
We distinguish two cases depending on the size of the slice $S$.
If $|S|=1$, then the new pivot $p_i$ is the unique element of the slice $S$.
If $|S|>1$, we say that we have {\em ties} and, in the general version of the SFS algorithm, we break them arbitrarily. We will see in Section \ref{sec:4-SFS+ algorithm} a variant of SFS (denoted by $\SFS_+$) where such ties are broken using a linear order given as additional input to the algorithm.
Once the new pivot $p_i$ is chosen, we mark it as visited (i.e.,  we remove it from the queue $\phi(p_{i-1})$) and we set $\sigma(p_i)=i$ (i.e., we let $p_i$ appear at position $i$ in~$\sigma$).

The second task is the update of the queue $\phi(p_{i-1})$, which
can be done as follows.
Intuitively, we update $\phi(p_{i-1})$ according to the similarities of $p_i$ with respect to the unvisited vertices and compatibly with the queue order.
Specifically, first we compute the similarity partition $\psi_{p_i}=(C_1,\ldots,C_s)$ of the neighborhood $N(p_i)$ of $p_i$ among the unvisited vertices
(see Definition \ref{def:layer partition}).
Second, we refine the ordered partition $\phi(p_{i-1})~\setminus~p_i=(B_1\setminus \{p_i\}, B_2,\ldots, B_r)$ by the ordered partition $\psi_{p_i}$ (see Definition~\ref{def:refine}). The resulting ordered partition is the ordered partition $\phi(p_i)$.

Note that if the matrix has only $0/1$ entries then the similarity partition $\psi_{p_i}$ has only one class, equal to the neighborhood  
of $p_i$ among the unvisited vertices.
Hence, the refinement procedure defined in Definition \ref{def:refine} simply reduces to the partition refinement operation defined in \cite{Habib00} for Lex-BFS.
This is why Lex-BFS is actually a special case of SFS for $0/1$ matrices. 

Note also that, by construction, each  class of the queue $\phi(p_i)$ is an interval of $\sigma$ (i.e., the elements of the class are consecutive in $\sigma$).
Furthermore, each of the  visited vertices $p_1, \dots, p_i$ is homogeneous to every class of the queue $\phi(p_i)$.

We show a simple example to illustrate how the algorithm works concretely.
Consider the following matrix:
\begin{equation*}
A=
\bordermatrix{
~ & \textbf{1} & \textbf{2} & \textbf{3} & \textbf{5} & \textbf{7} & \textbf{8} & \textbf{9} & \!\! \textbf{11} & \!\!\textbf{13} & \!\!\textbf{14}& \!\!\textbf{17} & \!\!\textbf{19}\cr
\textbf{1}  & * & 0 & 7 & 3 & 3 & 3 & 0 & 3 & 3 & 4 & 0 & 3 \cr
\textbf{2}  & & * & 0 & 7 & 6 & 3 & 8 & 3 & 3 & 0 & 8 & 6 \cr
\textbf{3}  & & & * & 3 & 3 & 3 & 0 & 3 & 3 & 8 & 0 & 3 \cr
\textbf{5}  & & & & * & 6 & 5 & 7 & 5 & 5 & 3 & 7 & 8 \cr
\textbf{7}  & & & & & * & 5 & 6 & 5 & 5 & 3 & 6 & 7 \cr
\textbf{8}  & & & & & & * & 4 & 8 & 6 & 5 & 4 & 5 \cr
\textbf{9}  & & & & & & & * & 4 & 3 & 0 & 8 & 6 \cr
\textbf{11} & & & & & & & & * & 7 & 5 & 4 & 5 \cr
\textbf{13} & & & & & & & & & * & 5 & 3 & 5 \cr
\textbf{14} & & & & & & & & & & * & 0 & 3 \cr
\textbf{17} & & & & & & & & & & & * & 6 \cr
\textbf{19} & & & & & & & & & & & & * \cr
}
\end{equation*}
studied in  \cite{Prea14} (we use also their original names for the vertices). 
In Figure \ref{fig:SFS sweep} are reported all the iterations of the SFS algorithm using as initial order of the vertices the reversal of the original labeling of the matrix.
At each iteration, the vertices in the blocks are the univisited vertices in the queue.

\begin{figure}[!ht]
\resizebox{\textwidth}{!}{
\begin{tabular}{C{\textwidth}}
\includegraphics[scale=0.95]{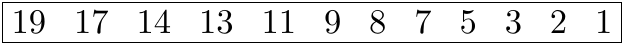} \hfill
\includegraphics[scale=0.95]{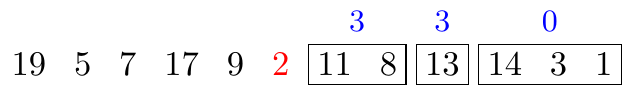} \hfill
\includegraphics[scale=0.95]{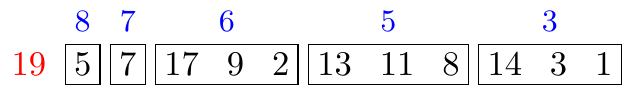} \hfill
\includegraphics[scale=0.95]{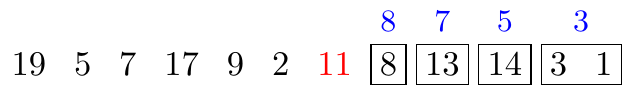} \hfill
\includegraphics[scale=0.95]{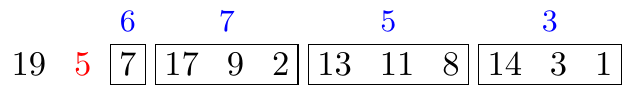} \hfill 
\includegraphics[scale=0.95]{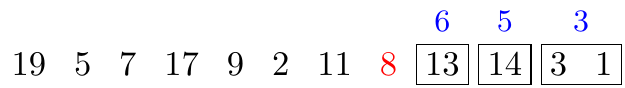} \hfill
\includegraphics[scale=0.95]{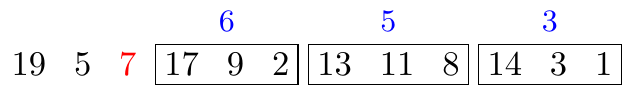} \hfill
\includegraphics[scale=0.95]{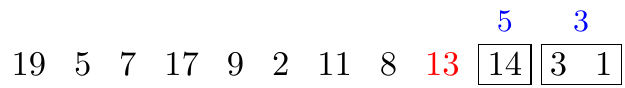} \hfill
\includegraphics[scale=0.95]{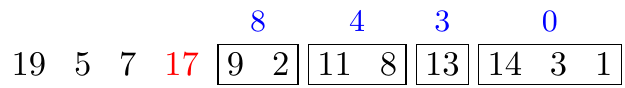} \hfill
\includegraphics[scale=0.95]{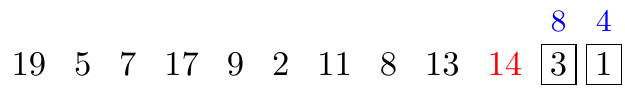} \hfill
\includegraphics[scale=0.95]{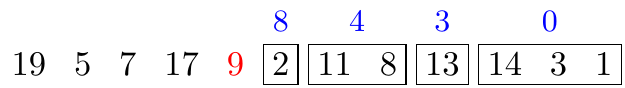} \hfill
\includegraphics[scale=0.95]{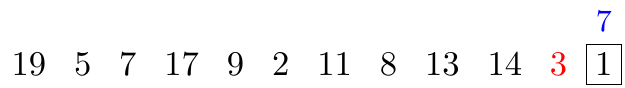} \hfill
\end{tabular} 
}
\caption{Iterations of SFS algorithm: in bold the pivot which is chosen at the current iteration, above the blocks the similarity between the new pivot and the vertices in the classes of the queue. The first line on the left shows the initialization step of the algorithm.}
\label{fig:SFS sweep}
\end{figure}

\subsection{Characterization of SFS orderings} \label{sec:3-characterization SFS}

In this subsection we characterize the linear orders returned by the SFS algorithm in terms of a 3-point condition.
This characterization applies to any (not necessarily Robinsonian) matrix and it is the analogue of  \cite[Thm 3.1]{Corneil09} for  Lex-BFS.

\begin{theorem}\label{thm:SFS ordering characterization}
Given a matrix $A\in \mathcal S^n$, an ordering  $\sigma$ of $[n]$ is a SFS ordering of $A$ if and only if the following condition holds: 
\begin{equation}\label{eq:SFS characterization}
\begin{array}{l}
\text{For all } x,y,z \in [n] \text{ such that } A_{xz}>A_{xy} \text{ and } x <_{\sigma} y <_{\sigma} z,\\
\text{there exists } u \in [n] \text{ such that } u <_{\sigma} x \text{ and }  A_{uy}>A_{uz}.
\end{array}
\end{equation}
\end{theorem}

\begin{proof}
$(\Rightarrow)$ 
Suppose $\sigma$ is a SFS ordering of $A$. Assume $x<_\sigma y<_\sigma z$  and $A_{xz} > A_{xy}$,  but $A_{uz} \geq A_{uy}$ for each $u<_{\sigma} x$. Assume first that  $A_{uz} > A_{uy}$ for some $u<_{\sigma} x$ and let $u$ be the first such vertex in $\sigma$. 
Then $A_{wz}=A_{wy}$ for each $w <_{\sigma} u$, and thus $y,z$ are in the same class of the queue of unvisited vertices when $u$ is chosen as pivot.
Therefore, $z$ would be ordered before $y$ in $\sigma$ when computing the similarity partition of $N(u)$, i.e., we would have  $z <_{\sigma} y$, a contradiction.
Hence,  one has $A_{uz}=A_{uy}$ for each $u<_{\sigma}x$. This implies  that $y,z$ are in the same class of the queue of unvisited vertices before $x$ is chosen as pivot. Hence, when $x$ is chosen as pivot, as $A_{xz} > A_{xy}$, when computing the similarity partition of $N(x)$ we would get $z <_{\sigma} y$, which is again a contradiction.

$(\Leftarrow)$ 
Assume that the condition (\ref{eq:SFS characterization}) of the theorem holds, but  $\sigma$ is not a SFS ordering. Let $a$ denote the first vertex of $\sigma$. Let   $\tau$ be a SFS ordering of $A$ starting at $a$ with  the largest possible initial overlap with $\sigma$. 
Say, $\sigma$ and $\tau$ share the same initial order $(a,a_1,\ldots,a_r)$ and they differ at the next position. Then we have that 
$\sigma=(a,a_1,\ldots,a_r,y, \ldots, z, \ldots,)$ and 
$\tau=(a,a_1,\ldots,a_r,z,\ldots, y,\ldots)$ with $y\ne z$.

\ignore{
\begin{figure}[!ht]
\centering
\includegraphics[page=1]{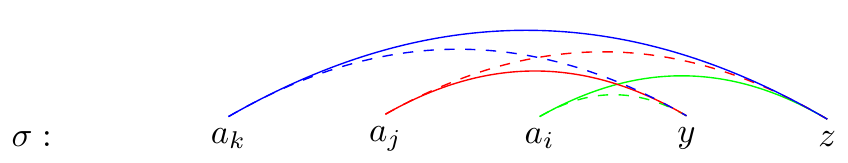}
\\
\includegraphics[page=2]{SFS_characterization.pdf}
\end{figure}
}

In the SFS ordering  $\tau$, the two elements $y,z$ do not lie in the slice of the pivot~$a_r$. Indeed, if $y,z$ would lie in the slice of $a_r$ then one could select $y$ as the next pivot instead of $z$, which would result in another SFS ordering  $\tau'$ starting at $a$ and with a larger overlap with $\sigma$ than $\tau$. 
Hence, there exists $i\le r$ such that $A_{a_iz}>A_{a_iy}$.  
Since $a_i <_{\sigma} y <_{\sigma} z$ then applying the condition (\ref{eq:SFS characterization}) to  $\sigma$, we deduce that there exists $j<i$ such that $A_{a_jy}>A_{a_jz}$.
Now, we have $a_j<_\tau z <_\tau y$ with $A_{a_jy}>A_{a_jz}$. As $\tau $ is a SFS ordering, as we have just shown it must satisfy the condition (\ref{eq:SFS characterization}) and thus there must exist an index $k<j$ such that
$A_{a_kz}>A_{a_ky}$.
Hence, starting from an index $i\le r$ for which $A_{a_iz}>A_{a_iy}$, we have shown the existence of another index $k<i \leq r$ for which $A_{a_kz}>A_{a_ky}$.
Iterating this process, we reach a contradiction. We will use in some other proofs this same type  of infinite chain argument, based on constructing an infinite chain of elements.
$\qquad$
\end{proof}

One can easily show that if $\sigma$ is a SFS ordering of $A$ and $S \subseteq V$ is a subset such that any element $x \notin S$ is homogeneous to~$S$, then the restriction $\sigma[S]$ of $\sigma$ to~$S$ is a SFS ordering of $A[S]$.
Note that, by construction, if we consider a generic slice $S$ encountered during the execution of  the SFS algorithm returning~$\sigma$,
then each vertex coming before $S$ in $\sigma$ is homogeneous to $S$. 
Hence,  a direct consequence of Theorem~\ref{thm:SFS ordering characterization} is that the restriction of $\sigma$ to any slice $S$ encountered throughout  the  SFS algorithm returning $\sigma$  is a SFS ordering of the  submatrix $A[S]$.

\subsection{The Path Avoiding Lemma} \label{sec:3-PAL}

In this subsection we discuss a fundamental lemma which we call the `Path Avoiding Lemma'.
It will play a crucial role throughout the paper and, in particular, for the characterization of anchors. 
Differently from the analysis in the previous subsection, where we did not make any assumption on the structure of the matrix  $A$, the Path Avoiding Lemma states some important properties of  SFS orderings when the input matrix is Robinsonian.

\medskip
Before stating this lemma, we need to investigate in  more detail the refinement step in the SFS algorithm.
An important operation in the Refine task in Algorithm \ref{alg:SFS} is the splitting procedure of each class of the queue $\phi$.
The following notion of `vertex splitting a pair of vertices' is useful to understand it.
Consider an order $\sigma=\SFS(A)$ and vertices $x<_\sigma y<_\sigma z$, where  $x=p_i$ is  the pivot chosen at the $i$th iteration  in Algorithm \ref{alg:SFS}.
We say that $x$ {\em splits} $y$ and $z$ if $x$ is the first pivot for which $y$ and $z$ do not belong to the same class in the queue ordered partition $\phi({p_i})$.
Recall that~$\phi(p_i)$ denotes the queue  of unvisited nodes induced by pivot~$p_i$, i.e., at the end of iteration~$i$ (after the refinement step).
Hence, saying that  $y,z$ are split by $x$ means that   $y,z$ belong to a common class $ B_j$  of $ \phi(p_{i-1})$ and that  they belong to distinct classes $B_h,B_k$ of $\phi(p_i)$,  where $ y \in B_{h},$ $z \in B_k$ and   $B_h$ comes before $ B_k$ in $\phi(p_i)$.
Equivalently, $x=p_i$ splits $y$ and $z$ if $A_{xy} > A_{xz}$ and  $A_{uy}=A_{uz}$ for all  $u<_\sigma p_i$.

Then, we say that two vertices $y<_\sigma z$ are {\em split} in $\sigma$ if they are split by some vertex $x<_\sigma y$. 
When $y$ and $z$ are not split in $\sigma$,  we say that they are {\em tied}. In this case, ties must be broken between $y$ and $z$.
In the SFS algorithm we  assume that ties are broken arbitrarily. In Section \ref{sec:4-SFS+ algorithm} we will see the variation  $\SFS_+$ of SFS  where ties are broken using a linear order $\tau$ given as input together with the matrix $A$.
The following lemma will be used as base case for proving the Path Avoiding Lemma.

\begin{lemma}\label{thm:equality A_xz=A_xy}
Assume that $A \in \MS^n$ is a  Robinsonian matrix and let  $\sigma=\SFS(A)$. Assume  that $x <_{\sigma} y <_{\sigma} z$ and that there exists a Robinson ordering $\pi$ of $A$ such that $x <_{\pi} z <_{\pi} y$. Then $y$ and $z$ are not split in $\sigma$ by any vertex $u \leq_{\sigma} x$. That is,
$A_{uy}=A_{uz}$ for all $u\le_\sigma x$. 
\end{lemma}

\begin{proof}
We first show that $y,z$ are not split by any vertex $w$ occurring before $x$ in $\sigma$. 
Suppose, for contradiction, that $y,z$ are split by a vertex $w <_\sigma x$. Hence, $A_{wy} > A_{wz}$.
This implies $z <_\pi w$ for, otherwise, $w<_\pi z<_\pi y$ would imply $A_{wy} \le A_{wz}$, a contradiction.
Hence we have $w <_{\sigma} x <_{\sigma} z$ and $x <_{\pi} z <_{\pi} w$.
Because $\pi$ is a Robinson ordering, we get $A_{wz} \geq A_{wx}$ and thus $A_{wy} > A_{wz} \geq A_{wx}$.
Therefore, the quadruple  $(w,x,y,z)$  satisfies the following properties (a)-(d): 
(a) $w<_\sigma x<_\sigma y<_\sigma z$, 
(b)  $x<_\pi z <_\pi w$  for some Robinson ordering $\pi$, 
(c) $w$ is the pivot splitting $y,z$, 
and (d) $A_{wy}>A_{wx}, A_{wz}$ 
Call any quadruple satisfying (a)-(d)  a {\em bad quadruple}.

\begin{figure}[!ht]
\centering
\includegraphics{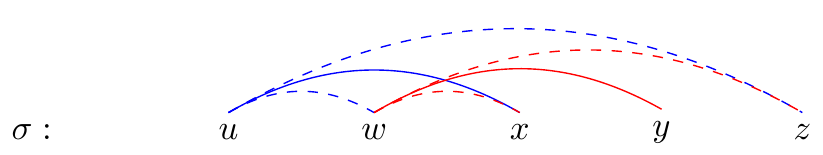}
\caption{Illustrating the proof of Theorem \ref{thm:equality A_xz=A_xy}: $(w,x,y,z)$ and $(u,w,x,z)$ are bad quadruples (the dotted lines indicate similarities that are strictly smaller than the continuous ones of the same thickness)}
\end{figure}

We now show that if $(w,x,y,z)$ is a bad quadruple then there exists $u<_\sigma w$ for which $(u,w,x,z)$ is also a bad quadruple. Hence, iterating we will get a contradiction (so we use here too an infinite chain argument).
We now proceed to show the existence of $u<_\sigma w$ for which $(u,w,x,z)$ is also a bad quadruple.
Since $A_{wx} < A_{wy}$,  the vertices $x,y$ are already split before $w$ becomes a pivot; otherwise, if they would belong  to the same class when $w$ is chosen as new pivot, then we would get $y <_{\sigma} x$.
Let $u=p_i$ the pivot splitting $x,y$, i.e., $u <_{\sigma} w$ and $A_{ux} > A_{uy}$. 
Thus $x,y$ belong to the same class (say) $B \in \phi(p_{i-1})$ when $u$ is chosen as new pivot at iteration $i$, but in different classes of $\phi(p_{i})$.
Since $w$ is the  pivot splitting $y,z$ and $u<_\sigma w$,  it follows that $y,z$ belong to the same class when $u$ is chosen as pivot, and thus  $x,y,z\in B$. 
Therefore $u$ is also the pivot splitting $x$ and $z$ and thus $A_{ux}>A_{uy}=A_{uz}$. 
In turn this implies that $u<_\pi z$ for, otherwise, $x<_\pi z<_\pi u$ would imply $A_{ux}\le A_{uz}$, a contradiction.
Therefore, $u<_\pi z<_\pi w$ and by definition of Robinson ordering we have $A_{uw}\le A_{uz}$  and, as $A_{ux}>A_{uz}$, this implies that $A_{uw} < A_{ux}$.
Summarizing, we have shown that the quadruple $(u,w,x,z)$ is bad since  it satisfies the conditions (a)-(d): (a) $u<_\sigma w<_\sigma x<_\sigma z$, (b) $w<_{\pi^{-1}} z <_{\pi^{-1}}  u$ for the Robinson ordering $\pi^{-1}$,  (c) $u$ splits $x$ and $z$, and (d) $A_{ux}>A_{uw}, A_{uz}$.
Thus we have shown  that  there cannot exist a bad quadruple and therefore that $y,z$ are not split by any vertex $w$ appearing before $x$ in $\sigma$.

We now conclude the proof of the lemma by showing that $y,z$ are also not split by $x$.
For this, we need to show that $A_{xz}=A_{xy}$. 
Suppose for contradiction that $A_{xz} \neq A_{xy}$. As $x<_\pi z<_\pi y$,  it can only be that $A_{xz} > A_{xy}$. 
Let $x=p_i$, i.e., $x$ is the pivot chosen at iteration $i$ of Algorithm \ref{alg:SFS}.
Since we have just shown that $y,z$ are not split before $x$, then at the iteration $i$ when $x$ is chosen as pivot, we would order $z <_{\sigma} y$ as $A_{xz} > A_{xy}$, which is a contradiction because $y<_{\sigma} z$ by assumption.
$\qquad$
\end{proof}

A first direct consequence of Lemma \ref{thm:equality A_xz=A_xy} is the following.

\begin{corollary}\label{thm:anchor simplicial sigma + base case PAL}
Let $A \in \MS^n$ be a Robinsonian matrix, let $\sigma=SFS(A)$, and consider distinct elements $x,y,z \in V$ such that  $x<_\sigma y<_\sigma z$. 
The following holds:
\begin{itemize}
\item[(i)] 
 $A_{xy}\ge \min\{A_{xz},A_{yz}\}$. 
 \item[(ii)] 
 If $x<_\pi z<_\pi y$ for some Robinson ordering $\pi$, then the path $P=(x,z)$ does not avoid $y$.
 \end{itemize}
 \end{corollary}
\begin{proof}
(i) Assume, for contradiction, that $A_{xy} < \min \{A_{xz},A_{yz}\}$. 
Pick a Robinson ordering $\pi$ of $A$ such that $x<_\pi y$. Then we must have $x <_{\pi} z <_{\pi} y$. Indeed,
 if $x <_{\pi} y <_{\pi} z$ then we would have $A_{xy} \geq A_{xz}$, and if $z <_{\pi} x <_{\pi} y$ we would have $A_{xy} \geq A_{yz}$, leading in both cases to a contradiction. 
Applying Lemma \ref{thm:equality A_xz=A_xy}, we conclude that $A_{xy}=A_{xz}$,  contradicting our assumption that $A_{xy}<A_{xz}$.

(ii) If $(x,z)$ avoids $y$ then $A_{xz} > \min \{A_{xy},A_{yz})$, where $\min\{A_{xy},A_{yz})= A_{xy}$ since $x<_\pi z<_\pi y$. Hence this contradicts  Lemma~\ref{thm:equality A_xz=A_xy}.
$\qquad$
\end{proof}

Note that the above result is  the analogue of the `$P_3$-rule'  for chordal graphs 
in~\cite[Thm 3.12]{Corneil09},  
which claims that, for any distinct $x,y,z\in V$ such that  $x<_\sigma y<_\sigma z$ while $x<_\pi z<_\pi y$ for some Robinson ordering $\pi$, the path $(x,z)$ does not avoid $y$. The next lemma strengthens the result of Corollary \ref{thm:anchor simplicial sigma + base case PAL} (ii), by showing  that there cannot exist {\em any} path from $x$ to $z$ avoiding $y$ and appearing fully before $z$ in $\sigma$.
We will refer to Lemma~\ref{thm:PAL}  below as the `Path Avoiding Lemma', also abbreviated as~(PAL) for ease of reference in the rest of the paper.

\begin{lemma}[\textbf{Path Avoiding Lemma (PAL)}]\label{thm:PAL}
Assume $A\in \mathcal S^n$ is a Robinsonian matrix and let $\sigma=\SFS(A)$. Consider distinct elements $x,y,z \in V$ such that  $x<_\sigma y<_\sigma z$. 
If $x<_\pi z<_\pi y$ for some Robinson ordering $\pi$, then there does not exist a path 
$P=(x,u_1,\ldots,u_k,z)$ from $x$ to $z$ avoiding $y$ and such that $u_1,\ldots,u_k<_\sigma z$.
\end{lemma}

\begin{proof}
The proof is by induction on the length $|P|=k+2$ of the path $P$.
The base case is $|P|=2$, i.e., $P=(x,z)$, which  is settled by Corollary \ref{thm:anchor simplicial sigma + base case PAL}.
Assume then, for contradiction,  that there exists a path $P=(x,u_1,\ldots,u_k,z)$ from $x$ to $z$ avoiding~$y$ with $u_1,\ldots,u_k<_\sigma z$ and $|P|\ge 3$, i.e., $k \geq 1$.
Let us call a path $Q$ {\em short} if it is shorter than $P$, i.e., if $|Q|<|P|$.
By the induction assumption, we know that the following holds: 
\begin{equation}\label{eqH}
\begin{array}{l}
\text{If } u<_\sigma v<_\sigma w \text{ and }
u<_\tau w<_\tau v \text{ for some Robinson ordering } \tau,\\
\text{then no short path } Q=(u,v_1,\ldots,v_r,w) \text{ from } u \text{ to } w \text{ avoiding } v\\
\text{and with } v_1,\ldots,v_r<_\sigma w \text{ exists}.
\end{array}
\end{equation}

Set $u_0=x$ and $u_{k+1}=z$.
As $P$ avoids $y$, the following relations hold:
\begin{equation}\label{eq0}
A_{u_{i-1}u_i}>\min\{A_{yu_{i-1}}, A_{yu_i}\} \ \text{ for all } i\in [k+1].
\end{equation}

Recall that since $x <_{\sigma} y <_{\sigma} z$ and $x <_{\pi} z <_{\pi} y$, then in  view of Lemma \ref{thm:equality A_xz=A_xy} we have $A_{xy}=A_{xz}$.
Furthermore, we know that $u_1,\dots,u_k<_{\sigma} z$ by  assumption. 
In order to conclude the proof, we use the following claim.

\begin{claim}\label{claim2}
$u_i<_\pi x$ and $y<_\sigma u_i$ for each $i \in [k]$.
\end{claim}
\begin{proof}
The proof is by induction on $i\ge 1$. For $i=1$  we have to show that 
\begin{equation}\label{eq1}
u_1<_\pi x \text{ and } y<_\sigma u_1.
\end{equation}

We first show that $u_1<_\pi x$. Suppose this is not the case and $x<_\pi u_1$.
Recall that in view of (\ref{eq0}) for $i=1$ we have $A_{xu_1} > \min \{ A_{yx}, A_{yu_1} \}$ and thus the path $(x,u_1)$ avoids $y$.
Hence, since $x <_{\pi} y$ and $y$ cannot appear between $x$ and $u_1$ in any Robinson ordering in view of Lemma \ref{thm:path avoiding and Robinson ordering}, it must also be that $u_1 <_{\pi} y$.
We then have two possibilities, depending whether $u_1$ comes before or after $z$ in $\pi$.

\begin{enumerate}[(i)]
\item 
Assume first that $u_1$ appears before $z$ in $\pi$. Then  we have $x<_\pi u_1<_\pi z<_\pi y$. We discuss where can $u_1$ appear  in $\sigma$.
If $u_1<_\sigma y$  then  we have  $u_1<_\sigma y<_\sigma z$,  $u_1 <_\pi z<_\pi y$, and  
$(u_1,\ldots,u_k,z)$ is a short path from $u_1$ to $z$ avoiding $y$ with $u_2,\dots,u_k <_{\sigma} z$, which contradicts~(\ref{eqH}). 
Hence, $y<_\sigma u_1$ in which case  we have
$x<_\sigma y<_\sigma u_1$, $x<_\pi u_1<_\pi y$, and  $(x,u_1)$ is a short path from $x$ to $u_1$ avoiding~$y$, which contradicts again (\ref{eqH}).
\item 
Assume now that $u_1$ appears after $z$ in $\pi$. Then we have $x<_\pi z<_\pi u_1<_\pi y$. By (\ref{eq0}) applied to $i=1$ and using the Robinson ordering $\pi$, we have that $A_{u_1x}>\min \{A_{yx},A_{yu_1}\}= A_{yx}$. Recall that $A_{xy} = A_{xz}$. Then $A_{u_1x} > A_{xz}$.
On the other hand, by the Robinson property of $\pi$, $A_{xu_1}\le A_{xz}$, yielding a contradiction.
\end{enumerate}
\smallskip
Therefore we have shown that $u_1<_\pi x$. Finally, we show that $y<_\sigma u_1$. Suppose not, i.e., $u_1<_\sigma y$. Then we would have
$u_1<_\sigma y <_\sigma z$ and, as just shown, $u_1<_\pi z <_\pi y$, while $(u_1,\ldots,u_k,z)$ is a short path from $u_1$ to $z$ avoiding $y$ with $u_2,\dots,u_k <_{\sigma} z$. 
This contradicts (\ref{eqH}) and thus shows $y<_\sigma u_1$, which concludes the proof for the base case $i=1$.

\medskip

Assume now that $i \geq 2$ and that $u_j<_\pi x$ and $y<_\sigma u_j$ for all 
$1\le j\le i-1$ by induction. We show that also $u_i<_\pi x$ and $y<_\sigma u_i$. First we show $u_i<_\pi x$. 
Suppose, for the sake of contradiction, that $x<_\pi u_i$.
Recall that in view of (\ref{eq0}) the path $(u_i,\dots,u_k,z)$ is a path from $u_i$ to $z$ avoiding $y$ with $u_{i+1},\dots,u_{k} <_{\sigma} z$.
Hence, since $z <_{\pi} y$ in view of Lemma \ref{thm:path avoiding and Robinson ordering} it must be also $u_i <_{\pi} y$, because $y$ cannot appear between $z$ and $u_i$ in any Robinson ordering.
We then have two possibilities to discuss, depending whether $u_i$ comes before or after $z$ in~$\pi$.

\begin{enumerate}[(i)]
\item 
Assume 
that $u_i$ appears before $z$ in $\pi$.
Then 
$u_1,\ldots, u_{i-1}<_\pi x <_\pi u_i <_\pi z<_\pi y$.
First we claim that $y<_\sigma u_i$. 
Indeed, if by contradiction $u_i<_\sigma y$, then we would have: $u_i<_\sigma y<_\sigma z$ and $u_i<_\pi z<_\pi y$, while $(u_i,\ldots,u_k,z)$ is a short path from $u_i$ to $z$ avoiding $y$ with $u_{i+1},\dots,u_{k} <_{\sigma} z$, contradicting (\ref{eqH}). 

Hence, $y<_\sigma u_i$ holds. Recall that $y <_{\sigma} u_j$ for $j\in[i-1]$ by induction. Hence, for $j=i-1$ we have $y <_{\sigma} u_{i-1}$. 
To recap, we are therefore in the case $u_{i-1} <_{\pi} x <_{\pi} u_{i} <_{\pi} z <_{\pi} y$ and we have shown that $x <_{\sigma} y <_{\sigma} u_{i}, u_{i-1} <_{\sigma} z$.

We thus have $y<_\sigma u_{i-1} <_\sigma z$ and  $y<_{\pi^{-1}} z <_{\pi^{-1}} u_{i-1}$. Then, in view of Lemma \ref{thm:equality A_xz=A_xy}, one must have 
$A_{yu_{i-1}}=A_{yz}$. From the Robinson ordering we obtain $A_{yz} \geq A_{xy} \geq A_{yu_{i-1}} = A_{yz}$ and therefore we get the equality $A_{yz} =A_{xy}$.
Analogously, because $x<_\sigma y<_\sigma u_i $ and $x<_\pi u_i <_\pi y$,  by Lemma \ref{thm:equality A_xz=A_xy} we obtain $A_{xy}=A_{xu_i}$.
Hence, we have 
\begin{equation}\label{eq:2}
A_{yu_{i-1}}=A_{yz}=A_{xy}=A_{xu_i}
\end{equation}
Finally, using relation (\ref{eq0}) we get: 
\begin{equation}\label{eq:3}
A_{u_{i-1}u_i} > \min \{A_{yu_{i-1}},A_{yu_i}\}= A_{yu_{i-1}}.
\end{equation}
In view of (\ref{eq:2}), the right hand side in (\ref{eq:3}) is $A_{yu_{i-1}}=A_{xu_i}$. 
On the other hand, as $u_{i-1}<_\pi x<_\pi u_i$ in the Robinson ordering $\pi$, then 
$A_{yu_{i-1}}=A_{xu_i} \geq A_{u_{i-1}u_i}$, which contradicts (\ref{eq:3}). Hence 
$u_i$ cannot appear before $z$ in $\pi$.

\item
Assume $u_i$ appears after $z$ in $\pi$.
Then $u_1,\ldots, u_{i-1}<_\pi x <_\pi z <_\pi u_i <_\pi y$.
Observe that the path $(x,u_1,\ldots,u_{i-1},z)$ is a short path from $x$ to $z$ with $u_1,\dots,u_{i-1} <_{\sigma} z$  and thus it cannot avoid $y$,  otherwise we would contradict (\ref{eqH}). Since the path $(x,u_1,\dots u_{i-1})$ avoids  $y$ (as it is a subpath of $P$), it follows that the path $(u_{i-1},z)$ does not  avoid $y$.
Hence   $A_{u_{i-1}z}\le \min \{A_{y u_{i-1}}, A_{yz}\}$ which, using the Robinson ordering $\pi$,  in turn implies 
$A_{u_{i-1}z} = A_{y u_{i-1}}$. 
Then, using relation (\ref{eq0}), we get: $A_{u_{i-1}u_i}>\min \{A_{yu_{i-1}},A_{yu_i}\}= A_{yu_{i-1}}$.
Now combining with   $A_{yu_{i-1}} = A_{u_{i-1}z}$, we get $A_{u_{i-1}u_i}> A_{u_{i-1}z}$ which is a contradiction, since 
  from the Robinson ordering $\pi$ one must have  $A_{u_{i-1}u_i} \leq A_{u_{i-1}z}$.
Therefore we have shown also that $u_i$ cannot appear after $z$ in $\pi$.
\end{enumerate}
\smallskip
In summary we have shown that $u_i<_\pi x$ as desired.
Finally we now show that $y<_\sigma u_i$. 
Indeed, if $u_i<_\sigma y$ then we would have:
$u_i<_\sigma y<_\sigma z$ and  $u_i <_\pi z<_\pi y$, while $(u_i,\ldots, u_k,z)$ is a short path from $u_i$ to $z$ avoiding $y$ with $u_{i+1},\dots,u_k <_{\sigma} z$, which contradicts (\ref{eqH}).
This concludes the proof of the claim.
$\qquad$
\end{proof}

We  can now conclude the proof of Lemma \ref{thm:PAL}.
By Claim \ref{claim2}  we have the following relations for any $i\in [k]$:
$x <_\sigma y<_\sigma u_i <_\sigma z$ and $y<_{\pi^{-1}} z <_{\pi^{-1}} x <_{\pi^{-1}} u_i$.  
By Lemma~\ref{thm:equality A_xz=A_xy}, this implies
$A_{yu_i}= A_{yz}$ for all $i\in [k]$ which, using the Robinson ordering~$\pi$,  in turn implies 
$A_{yu_i}=A_{yz}=A_{yx}$.
Now, use relation (\ref{eq0}) for $i=k+1$ to get the inequality 
$A_{u_kz}>\min \{A_{yu_k},A_{yz}\}= A_{yu_k}$.
Recall that in view of Lemma~\ref{thm:equality A_xz=A_xy}, we have that $A_{xy}=A_{xz}$.
Then as $A_{yu_i}=A_{yx}$ for all $i$,  the right hand side is equal to $A_{yu_k}=A_{xz}$ while, using the Robinson ordering $\pi$, the left hand side satisfies
$A_{u_kz}\le A_{xz}$, which yields a contradiction.
This concludes the proof of the lemma.
\hfill
\end{proof}

\subsection{End-vertices of SFS orderings} \label{sec:3-end points}

In this subsection we show some  fundamental properties of SFS orderings, using the results in Subsection \ref{sec:3-PAL}.
First we show that if $A$ is Robinsonian then the last vertex of a SFS ordering of~$A$ is an anchor of~$A$.  We will see later in Corollary  \ref{corSFSRobinson}
that conversely any anchor can be obtained as end-vertex of a  SFS ordering.

\begin{theorem}\label{thm:last vertex of SFS is an anchor}
Let A be a Robinsonian matrix and let $\sigma=\SFS(A)$. Then the last vertex of $\sigma$ is an anchor of $A$.
\end{theorem}

\begin{proof}
Let $z$ be the last vertex of $\sigma$; we show that $z$ is an anchor of $A$. In view of Theorem \ref{thm:anchor admissible}  it suffices to show that $z$ is valid. 
Suppose for contradiction that, for some $x\ne y\in V\setminus \{z\}$, there exist a path $P$ from $x$ to $z$ avoiding $y$ and a path $Q$ from $y$ to $z$ avoiding $x$.
We may assume without loss of generality that $x<_\sigma y<_\sigma z$.
Moreover, let $\pi$ be a Robinson ordering of $A$ such that $x<_\pi z$. 
Then, in view of Lemma \ref{thm:path avoiding and Robinson ordering}, we must have $x<_\pi z <_\pi y$, since $y$ must come either before or after both $x$ and $z$ (because of the path $P$) and $x$ must come before or after both~$y$ and~$z$ (because of the path $Q$). 
As $z$ is the last vertex, then $P <_{\sigma} z$ and thus we get a contradiction with Lemma~\ref{thm:PAL} (PAL).
\hfill
\end{proof}

The above result is the analogue of 
 \cite[Thm 4.5]{Corneil09} for Lex-BFS applied to interval graphs.
We now introduce the concept of `good SFS'.

\begin{definition}[{\bf Good SFS ordering}]\label{defgoodSFS}
We say that a SFS ordering $\sigma$ of $A$ is {\rm good}  if $\sigma$  starts with a vertex which is the end-vertex of some SFS ordering. 
\end{definition}

Note that the analogous definition in \cite{Corneil09} for Lex-BFS is stronger, as it requires the first vertex of each slice to be an end-vertex of the slice itself.
However, in our discussion  we do not need such a strong definition and the above notion of good SFS will suffice to show the overall correctness of the multisweep algorithm. In the case when $A$ is Robinsonian,  in view of Theorem \ref{thm:last vertex of SFS is an anchor} (and Corollary \ref{corSFSRobinson} below),
   $\sigma$ is a good SFS ordering precisely when it starts with  an anchor of $A$. 
For good SFS orderings we have the following stronger result for their end-vertices.

\begin{theorem}\label{thm:end points of good SFS are opposite anchor}
Let $A \in \MS^n $ be a Robinsonian matrix and let $\sigma$ be a good SFS ordering whose first vertex is $a$ and whose last vertex is $b$. Then $a,b$ are opposite anchors of $A$.
\end{theorem}

\begin{proof}
By assumption, $\sigma$ is a  good SFS ordering and thus its first vertex $a$ is an anchor of $A$. In  view of Theorem \ref{thm:last vertex of SFS is an anchor}, its last vertex $b$ is also an  anchor of $A$. 
Suppose, for the sake of contradiction, that $a$ and $b$  are not opposite anchors of~$A$. 
Then, in view of Theorem~\ref{thm:characterization opposite anchors}, there exists a vertex~$x \in V$ and a path~$P$ from~$a$ to~$b$ such that~$P$ avoids~$x$. 
Let $\pi$ be a Robinson ordering of $A$ starting with $a$ (which exists since $a$ is an anchor of $A$).
Using Lemma \ref{thm:path avoiding and Robinson ordering}  applied to the path $P$, we can conclude  that~$x$ cannot appear between $a$ and $b$ in any Robinson ordering, and thus we must have $a <_{\pi} b <_{\pi} x$. But then,  using Lemma \ref{thm:PAL} (PAL), there cannot exist a path from~$a$ to~$b$ avoiding~$x$ and appearing before~$b$ in~$\sigma$, which contradicts the existence of~$P$.
\hfill
\end{proof}

\section{The SFS$_+$ algorithm} \label{sec:4-SFS+ algorithm}

In this section we introduce the $\SFS_+$ algorithm. 
This is a variant of the standard SFS algorithm, 
and it is the analogue of the variant Lex-BFS+ of Lex-BFS introduced by Simon \cite{Simon91} in  the study of multisweep algorithms for interval graphs 
(although the multisweep algorithm itself in~\cite{Simon91} is actually flawed, see \cite{Corneil09} for more details).
The algorithm $\SFS_+$ will be the main ingredient in our multisweep algorithm for the recognition of Robinsonian matrices. 
It takes as input a matrix $A$ and a linear  order $\sigma$, and it returns another linear order $\sigma_+$.
After describing  $\SFS_+$, we will first present  its main properties, most importantly the fact that the $\SFS_+$ algorithm `flips anchors' when applied to a Robinsonian matrix $A$ and a good SFS order $\sigma$: if $\sigma$ starts at $a$ and ends at $b$, then $\sigma_+$ starts at $b$ and ends at $a$. 
We will also introduce the useful concept of `similarity layers' of a matrix, which will play a crucial role in the correctness analysis of our multisweep SFS-based  algorithm.

\subsection{Description of the SFS$_+$ algorithm} \label{sec:4-good SFS and end points}%


Consider again the SFS algorithm as described in Algorithm \ref{alg:SFS}  in Section \ref{sec:3-SFS algorithm}. The first main task is selecting the new pivot. In case of ties,   as done at Line \ref{alg: ties} of Algorithm \ref{alg:SFS}, the ties are broken arbitrarily (choosing any vertex  in the slice $S$). 
We now introduce a variant of $\SFS(A)$, which we denote by $\SFS_+ (A,\sigma)$. 
It takes as input a matrix $A\in \mathcal S^n$ and a linear order $\sigma$ of $V$, and it returns a new linear order~$\sigma_+$ of~$V$. 
In the $\SFS_+$ algorithm, the input linear order $\sigma$ is used  to break ties at Line \ref{alg: ties} in Algorithm \ref{alg:SFS}.
Specifically, among the vertices in the slice $S$ of the current iteration, we choose as new pivot the vertex appearing last in $\sigma$. 
Notice that a $\SFS_+$ ordering is still a SFS ordering and thus it satisfies all the properties discussed in Section \ref{sec:3-SFS algorithm}.

\ignore{
\medskip
\begin{algorithm}[H]\label{alg:SFS+}
\caption{\textit{\textit{}}$\SFS_+(A, \sigma)$}
\SetKwInput {KwIn}{input}
\SetKwInput {KwOut}{output}
\KwIn{a matrix $A \in \mathcal{S}^{n}$ and a linear order $\sigma$ of $[n]$}
\KwOut{a linear order $\sigma_+$ of $[n]$}
\vspace{2ex}
$\phi=(V)$\\
\For{$i=1,\dots,n$}{
	$S$ is the first class of $\phi$\\	
	choose $p$ as the vertex in $S$ appearing last in $\sigma$\\
	$\sigma_+(p)=i$ \\ 
	remove $p$ from $\phi$ \\ 
	$N(p)$ is the set of vertices $y\in \phi$ with $A_{py}>0$\\
	$\psi_p$ is the similarity partition of $N(p)$ with respect to $p$\\
	$\phi=$\textit{Refine} $(\phi,\psi_p)$
}
\Return: $\sigma_+$
\end{algorithm}
\medskip
}

If $A$ is a Robinsonian matrix and the input linear order~$\sigma$ is a SFS ordering, then the $\SFS_+$ ordering $\sigma_+$ has some important additional properties.
In fact, since in the beginning of the SFS algorithm all the vertices are contained in the  `universal' slice (i.e., the full ground set $V$),
 the order
$\sigma_+$ starts with the last vertex of $\sigma$, which in view of Lemma \ref{thm:last vertex of SFS is an anchor} is an anchor of $A$.
Therefore, in this case, $\sigma_+$  is a good SFS ordering by construction.
Furthermore, in view of Theorem \ref{thm:end points of good SFS are opposite anchor}, when $A$ is Robinsonian then the first and  last vertices of $\sigma_+$ are opposite anchors of $A$.
If the input linear order $\sigma$ is a good SFS ordering, then we have an even stronger property: the end-vertices of~$\sigma_+$ are the end-vertices of $\sigma$ but in reversed order. We call this the `anchors flipping property',   which is shown in the next theorem. This property  will be crucial in Section \ref{sec:5-the multisweep algorithm} when studying the properties of the multisweep algorithm.

\begin{theorem}[\textbf{Anchors flipping property}]\label{thm:flipping theorem}
Let $A \in \MS^n$ be a Robinsonian similarity, let $\sigma$ be a good SFS ordering of $A$ and $\sigma_+= \SFS_+(A,\sigma)$.
Suppose that $\sigma$ starts with $a$ and ends with $b$. Then $\sigma_{+}$ starts with $b$ and ends with $a$.
\end{theorem}

\begin{proof}
By definition of the $\SFS_+$ algorithm, the returned order  $\sigma_{+}$ starts with the last vertex $b$ of $\sigma$.
Hence, we only have to show that $a$ appears last in $\sigma_{+}$. Suppose, for the sake of contradiction, that $a$ is not last in $\sigma_{+}$ and let instead $y$ be the vertex appearing last in $\sigma_{+}$. Then we have $a<_{\sigma} y <_{\sigma} b$ and $b<_{\sigma_+} a <_{\sigma_+} y$. This implies that 
 $y$ and $a$ must be split in~$\sigma_{+}$. 
Indeed,  if $y$ and $a$ would be  tied in $\sigma_{+}$ then, as we use $\sigma$ to break ties and as $a<_{\sigma}y$, the vertex $y$ would be placed before $a$ in~$\sigma_{+}$, a contradiction.
Let thus $x<_{\sigma_+} a $ be the pivot splitting $a$ and $y$ in $\sigma_{+}$, so that  $A_{xa} > A_{xy}$. Then we have:
\begin{equation}\label{eq:path}
A_{xa} > \min \{A_{xy},A_{ya}\}.
\end{equation}
Hence the path $P=(x,a)$ avoids $y$. As $b$ is the first vertex of $\sigma_+$, we have:
\begin{equation*}
b <_{\sigma_+} x <_{\sigma_+} a <_{\sigma_+} y.
\end{equation*}
In view of Theorem \ref{thm:end points of good SFS are opposite anchor} applied to $\sigma$,  we know that $a$ and $b$ are opposite anchors of $A$. Therefore, there exists a Robinson ordering  $\pi$  starting with $a$ and ending with~$b$. 
In view of (\ref{eq:path}) and using Lemma \ref{thm:path avoiding and Robinson ordering},  $y$ cannot appear between $a$ and $x$ in any Robinson ordering and therefore we can conclude:
\begin{equation}
a <_{\pi} x <_{\pi} y <_{\pi} b.
\end{equation}
Consider now $\sigma$. We have that $a <_{\sigma} y <_{\sigma} b$. Where can $x$ appear in $\sigma$? Suppose $y <_{\sigma} x$. Then we would have  $a <_{\sigma} y <_{\sigma} x$ and $a <_{\pi} x <_{\pi} y$, and in view of Lemma~\ref{thm:PAL}~(PAL) there cannot exist a path from $a$ to $x$ avoiding $y$ and appearing before $x$ in $\sigma$, which is a contradiction as the path $P=(x,a)$ avoids $y$ in view of ($\ref{eq:path}$).
Hence, we must have:
\begin{equation}
a <_{\sigma} x <_{\sigma} y <_{\sigma} b.
\end{equation}
Therefore, starting from the pair $(a,y)$ satisfying $a<_{\sigma} y$ and $a<_{\sigma_+} y$, we have constructed a new pair $(x,y)$ satisfying $x<_{\sigma} y$ and $x<_{\sigma_+} y$, with  $x<_{\sigma_+} a$. Iterating this construction we  get an infinite sequence of such pairs, yielding a contradiction. (Here too we have used  an infinite chain  argument.)$\qquad$
\end{proof}

The flipping property of anchors is the analogue of  \cite[Thm 4.6]{Corneil09} for Lex-BFS.
An important consequence of this  property is that, if the linear order $\sigma$ given as input  is a Robinson ordering of $A$, then $\sigma_+= \SFS_+(a,\sigma)$ is equal to $\sigma^{-1}$,  i.e., the reversed order of~$\sigma$.

\begin{lemma}\label{thm:SFS+(a, pi) and reversed Robinson triple}
Let $A \in \MS^{n}$ be a Robinsonian matrix and let $\sigma, \tau$ be two SFS orderings of $A$. The following holds:
\begin{enumerate}[(\textit{i})]
\item[(i)] If $x<_{\tau} y <_{\tau} z$ and $z<_{\sigma} y <_{\sigma} x$ then the triple $(x,y,z)$ is Robinson.
\item[(ii)]  If $\tau$ is a Robinson ordering of $A$ and $\sigma= \SFS_+(A,\tau)$, then $\sigma=\tau^{-1}$.
\end{enumerate}
\end{lemma}

\begin{proof}
\textit{(i)}
Suppose for contradiction that the triple $(x,y,z)$ is not Robinson. 
Then we have $A_{xz} > \min \{A_{xy},A_{yz}\}$, and thus the path $(x,z)$ avoids $y$.
Let $\pi$ be a Robinson ordering of $A$ with (say) $x<_\pi y$. 
In view of Lemma \ref{thm:path avoiding and Robinson ordering},  $y$ cannot appear between $x$ and $z$ in any Robinson ordering and  
therefore  we have  $x<_\pi z<_\pi y$ or $z <_\pi x <_\pi y$. In both cases we get a contradiction with 
Lemma \ref{thm:PAL} (PAL) since $x<_\tau y<_\tau z$ and $z<_\sigma y <_\sigma x$.

\textit{(ii)}
Say $\tau$ starts at $b$ and ends at $a$. Then 
 $\sigma$ starts at $a$. 
Assume that $\sigma \ne \tau^{-1}$. Let $(a=x_0,x_1,\ldots,x_k)$ be the longest initial segment of $\sigma$ whose reverse $(x_k,\ldots,x_1,a)$ is the final segment of $\tau$, with $k\ge 0$.
Let $y$ be the successor of $x_k$ in $\sigma$. Then $y$ is not the predecessor of $x_k$ in $\tau$ (by maximality of $k$). 
Let $z$ be the predecessor of $x_k$ in $\tau$.
Then $a<_\sigma x_1<_\sigma \ldots <_\sigma x_k<_\sigma y <_\sigma z$ and 
$y <_\tau z< _\tau x_k<_\tau \ldots <_\tau x_1 <_\tau a$.
Hence, $y,z$ cannot be tied in $\sigma$ (otherwise we would choose $z $ before $y$ in $\sigma$  as $y <_{\tau} z$). 
Therefore, there must exist a vertex $x<_{\sigma} y$ such that $A_{xy} > A_{xz}$. Hence, $x=x_i$ for some $0\le i\le k$ and thus 
 $y<_\tau z<_\tau x$. As $\tau$ is a Robinson ordering this implies $A_{xy}\le A_{xz}$, a contradiction.
$\qquad$
\end{proof}
 
In other words, in a multisweep algorithm applied to a Robinsonian matrix, every triple of vertices appearing in reversed order in two distinct sweeps is Robinson. Moreover,
once a given sweep is a Robinson ordering, the next sweep will  remain a Robinson ordering (precisely the reversed order).
As direct application of Lemma~\ref{thm:SFS+(a, pi) and reversed Robinson triple}, we have the following characterization for Robinsonian matrices.

\begin{corollary}\label{thm:corollary alternative Robinson check}
Let $A \in \MS^n$, let $\tau$ be a SFS ordering of $A$ and let $\sigma =$\SFS$_+(A,\tau)$. 
Assume that $\sigma = \tau^{-1}$. Then $A$ is Robinsonian if and only if $\sigma$ is Robinson. 
\end{corollary}

We will see in Section~\ref{sec:6-complexity} how to exploit the above result to check if a given SFS ordering is a Robinson ordering during a multisweep algorithm.
Furthermore, combining Lemma \ref{thm:SFS+(a, pi) and reversed Robinson triple} with Theorem \ref{thm:last vertex of SFS is an anchor},
we obtain the following characterization for anchors.

\begin{corollary}\label{corSFSRobinson}
Let $A\in \mathcal S^n$ be a Robinsonian matrix. A vertex is an anchor of~$A$ if and only if it is the end-vertex of a SFS ordering of $A$.
\end{corollary}

\subsection{Similarity layers} \label{sec:4-similarity layers}

In this subsection we introduce the notion of `similarity layer structure' for a matrix $A\in \mathcal S^n$ and an element $a\in V$ (then called the {\em root}), 
which we will use later to analyze properties of the multisweep algorithm. 

Specifically, we define the following  collection   $\mathcal L=(L_0, L_1, \ldots, L_r)$ of subsets of~$V$, whose members are called the {\em (similarity) layers of $A$ rooted at $a$}, where $L_0=\{a\}$ and the next layers  $L_i$ are the subsets of $V$ defined recursively as follows:
\begin{equation}\label{eq:layers}
L_{i}=\{y \notin L_0\cup  \dots\cup  L_{i-1}: A_{xy} \ge  A_{xz}\  \forall x \in L_0\cup \dots \cup L_{i-1}, \ \forall  z \notin L_0\cup  \dots\cup L_{i-1}\}.
\end{equation}

Note that this notion of  similarity layers  can be seen as a refinement of the notion of BFS layers for graphs, which are obtained by layering the nodes according to their distance to the root.
Hence, the two concepts are similar but different.
We first show that this  layer structure defines  a partition of $V$  when  $A$ is a Robinsonian matrix and the root $a$ is an anchor of $A$.

\begin{lemma}\label{thm:layers structure}
Assume that  $A \in \MS^n$ is a Robinsonian matrix and that $a \in V$ is an anchor of $A$.
Consider  the similarity layer structure $\mathcal L=(L_{0}=\{a\}, L_1, \dots, L_r)$ of~$A$ rooted at~$a$, as defined  in (\ref{eq:layers}), where $r$ is the smallest index such that $L_{r+1}= \emptyset$.
The following holds:
\begin{enumerate}[(1)]
\item[(i)] If $y \in L_{i}, z \not\in L_0 \cup \ldots \cup L_i$ with $i\ge 1$, then there exists a path $P$ from $a$ to $y$ avoiding $z$. Moreover, any path  of the form 
$P=(a,a_1,\ldots, a_i=y)$, where $a_l \in L_l$ for  $1\le l\le i$, avoids $z$. 
\item[(ii)]   $V=L_0\cup L_1 \cup \ldots \cup L_r$.
\end{enumerate} 
\end{lemma}

\begin{proof}
(i)  Using the definition of the layers in (\ref{eq:layers}) we obtain that $A_{aa_1}>A_{az}$ and $A_{a_1a_2}>A_{a_1z}$, $\ldots, A_{a_{i-1}y}>A_{a_{i-1}z},$ which shows that the path $(a,a_1,\ldots,a_{i-1},a_i=~y)$ avoids $z$.

(ii)
Suppose $L_0,L_1,\ldots, L_r \ne \emptyset$,  $L_{r+1} =\emptyset$, but $V\ne U:=L_0 \cup \ldots \cup L_r$.
Consider an element $z_1\in V\setminus U$. As $z_1\not\in L_r$ (since this set is empty) there exist elements $x_1\in U$ and $z_2\not\in U$ such that $A_{x_1z_1}<A_{x_1z_2}$. Analogously, as $z_2\not\in L_r$ there exist elements $x_2\in U,z_3\not\in U$ such that 
$A_{x_2z_2}<A_{x_2z_3}$.
Iterating we find elements $x_i\in U$, $z_i\not\in U$ for all $i\ge 1$ such that $A_{x_iz_i}<A_{x_iz_{i+1}}$ for all $i$. 
At some step one must find one of the previously selected elements $z_i$, i.e., $z_j=z_i$ for some $i<j$.

 As $a$ is an anchor of $A$, there exists a Robinson ordering $\pi$ of $A$ starting at $a$. 
We first claim that $x_i<_\pi z_j$ for all $i,j$. This is clear if $x_i=a$. Otherwise, as $x_i\in U$ and  $z_j\not\in U$, it follows from (i) that there is a path from $a$ to $x_i$ avoiding $z_j$, which in view of Lemma \ref{thm:path avoiding and Robinson ordering} implies that $a\le _\pi x_i < _\pi z_j$. 
Next we claim that $z_{i+1}<_\pi z_i$. 
Since $x_i<_\pi z_i$ and $A_{x_iz_i}<A_{x_iz_{i+1}}$, then $(x_i,z_{i+i})$ avoids $z_i$ and in view of Lemma \ref{thm:path avoiding and Robinson ordering} it must be indeed $z_{i+1}<_\pi z_i$.
Summarizing we have shown that $z_{i+1}<_\pi z_i<_\pi \ldots <_\pi z_1$ for all $i$, which contradicts the fact that  two of the $z_i$'s should coincide.
$\qquad$
\end{proof}

Intuitively, each layer $L_{i}$ will correspond to some slices of a SFS algorithm starting at $a$.
As we see below, there is some compatibility between the layer structure $\mathcal L$ rooted at $a$ with any Robinson ordering  $\pi$ and  any good SFS  ordering  $\sigma$ starting at $a$.

\begin{lemma}\label{thm:layers and SFS/Robinson ordering}
Assume  $A \in \MS^n$ is a Robinsonian matrix and $a$ is an anchor of $A$. 
Let $\sigma$ be a good SFS ordering of $A$ starting at $a$ and let $\pi$ be a  Robinson ordering of $A$ starting at $a$.
Then the similarity layer structure $\mathcal L=(L_0=\{a\}, \dots, L_r)$ of $A$ rooted at  $a$ is compatible with both $\pi$ and $\sigma$. That is,
\begin{align*}
L_0<_\pi L_1 <_\pi \ldots <_\pi L_r, \\ 
L_0<_{\sigma} L_1 <_{\sigma} \ldots <_{\sigma} L_r.
\end{align*}
\vspace{-2ex}
\end{lemma}
\begin{proof}
Let $x\in L_i$ and $y\in L_j$ with $i<j$; we show that $x<_\pi y$ and $x<_{\sigma} y$.
This is clear if $i=0$, i.e., if $x=a$. 
Suppose now $i\ge 1$. 
Then, by Lemma \ref{thm:layers structure}, there exists a path from $a$ to $x$ avoiding $y$.
This implies that $a<_\pi x<_\pi y$, as $y$ cannot appear between $a$ and $x$ in any Robinson ordering in view of Lemma \ref{thm:path avoiding and Robinson ordering} and since $\pi$ starts with $a$.
Furthermore, if $a<_{\sigma} y<_{\sigma} x$ then we would get a contradiction with Lemma~\ref{thm:PAL}~(PAL). Hence  $a<_{\sigma} x <_{\sigma} y$ holds, as desired.
$\qquad$
\end{proof}

Furthermore, the following inequalities hold among the entries of $A$ indexed by  elements in different layers.

\begin{lemma}\label{thm:inequalities among layers}
Assume  $A \in \MS^n$ is a Robinsonian matrix and~$a$ is an anchor of~$A$.
Let  $\mathcal L=(L_0=\{a\}, L_1,\dots, L_r)$ be  the similarity layer structure of $A$ rooted at  $a$.
For each $u\in L_i, x,y\in L_j$ and $z\not\in L_0\cup L_1\cup  \ldots \cup L_j$ with $0 \leq i<j$ the following inequalities hold:
\begin{equation*}
A_{xy} \geq A_{ux}=A_{uy} \geq A_{uz}.
\end{equation*}
Furthermore, if $x \in L_j, z \notin L_0 \cup L_1\cup  \dots \cup L_j$, then there exists $u \in L_0 \cup L_1\cup  \dots \cup L_{j-1}$ such that $A_{ux}>A_{uz}$.
\end{lemma}

\begin{proof}
The inequalities $A_{ux}=A_{uy}>A_{uz}$ follow from the definition of the layers in (\ref{eq:layers}). 
Suppose now that $A_{xy} < A_{ux}=A_{uy}$. Then $u$ must appear between $x$ and $y$ in any Robinson ordering $\pi$, since $x <_{\pi} y <_{\pi} u$ implies $A_{xy} \geq A_{ux}$ and $y <_{\pi} x <_{\pi} u$ implies $A_{xy} \geq A_{uy}$.
But in view of Lemma \ref{thm:layers and SFS/Robinson ordering},  if $\pi$ is a  Robinson ordering starting at $a$ then $u <_{\pi} x$ and $u <_{\pi} y$, so we get a contradiction. 
$\qquad$
\end{proof}

As an application of  Lemma \ref{thm:inequalities among layers}, it is easy to verify  that if $A$ is the  adjacency matrix of a connected graph $G$, then each layer is a clique of $G$. 

We now show a `flipping property' of 
 the similarity layers with respect  to  a good SFS ordering $\sigma$ starting at the root and the next sweep $\sigma_+=\SFS_+(A,\sigma)$.
Namely we show that  the orders of the layers are reversed beween $\sigma$ and $\sigma_+$, i.e., $L_i<_\sigma L_j$ and $L_j<_{\sigma_+} L_i$ for all $i<j$.

\begin{theorem}[\textbf{Layers flipping property}]\label{thm:flipping layers theorem}
Let  $A \in \MS^n$ be a Robinsonian~matrix and $a\in V$ be an anchor of $A$.
Let $\mathcal L=(L_0=\{a\}, \dots, L_r)$ be the similarity layer structure of $A$ rooted at $a$, 
let $\sigma$ be a good SFS ordering of $A$ starting at $a$  and let $\sigma_+=\SFS_+(A,\sigma)$.
If  $x\in L_i$, $y\in L_j$ with $0\le i<j\le r$ then $y<_{\sigma_+} x$.
\end{theorem}

\begin{proof}
Let $x\in L_i$, $y\in L_j$ with $i<j$. Assume for contradiction that $x<_{\sigma_+} y$.
By Lemma~\ref{thm:layers and SFS/Robinson ordering}, we know that $\mathcal{L}$ is compatible with $\sigma$ and thus $x<_{\sigma} y$.
As $x<_{\sigma_+} y$ and $x<_{\sigma} y$, we deduce that $x,y$ are not tied in $\sigma_{+}$. 
Hence there exists $x_1<_{\sigma_+} x$ such that $A_{x_1x}>A_{x_1y}$.
Let $L_\ell$ denote the layer of $\mathcal{L}$ containing $x_1$. We claim that $\ell<j$.
Indeed, if $\ell=j$ then $x_1,y$ are in the same layer and, by Lemma \ref{thm:inequalities among layers}, it must be $A_{x_1y} \geq A_{x_1x}=A_{xy}$ which is impossible, because $A_{x_1x} > A_{x_1y}$. Assume now that  $\ell>j$. By Lemma~\ref{thm:layers and SFS/Robinson ordering}, if $\pi$ is a Robinson ordering starting at $a$,
then  we would get $x <_{\pi} y <_{\pi} x_1$, which implies $A_{x_1y} \geq A_{x_1x}$, again a contradiction.
Therefore, we have  $x_1 \in L_\ell$ with $\ell < j$. Recall that $x_1 <_{\sigma_+} y$.
Hence, starting with the pair $(x,y)$ which satisfies $x\in L_i$, $y\in L_j$ with $i<j$ and $x<_{\sigma_+} y$, we have constructed another pair $(x_1,y)$ 
satisfying $x_1\in L_l$, $y\in L_j$ with $l<j$ and $ x_1<_{\sigma_+} y$. As $x_1<_{+} x$,  iterating this construction we will reach a contradiction.
$\qquad$
\end{proof}

\section{The multisweep algorithm}\label{sec:5-the multisweep algorithm}

We now  introduce our new SFS-based  multisweep  algorithm and we show that in at most $n-1$ sweeps it permits to recognize whether a given matrix of size $n$ is Robinsonian. This  is the  main result of our paper,  which we will prove in this section.
First in  Subsection \ref{sec:5-description} we will describe the algorithm and its main features. Then in Subsection \ref{sec:5-three-good SFS ordering} we introduce the notion of `3-good sweep'   which plays a crucial role   in the correctness proof and we investigate its properties. 
In Subsection~\ref{sec:5-final proof} we complete the proof of correctness of the multisweep algorithm.
Finally, in Subsection~\ref{sec:5-worst case instances} we present an infinite family of $n \times n$ Robinsonian matrices  whose recognition needs exactly $n-1$ sweeps.

\subsection{Description of the multisweep algorithm} \label{sec:5-description}

Our multisweep algorithm  consists of computing successive SFS orderings of a given nonnegative matrix $A\in \mathcal S^n$.
The first sweep is $\SFS(A)$, whose aim is to find an anchor of $A$.  Each subsequent sweep  is computed with the $\SFS_+$ algorithm using the linear order returned by the preceding sweep to break ties.
\ignore{(as in Algorithm \ref{alg:SFS+}).}  
As it starts with the end-vertex of the preceding sweep which is an anchor of $A$, each subsequent sweep is therefore a good SFS ordering of $A$ (in the case when $A$ is Robinsonian).
The algorithm terminates either if a Robinson ordering has been found (in which case it certifies that $A$ is Robinsonian), 
 or if the $(n-1)$th sweep  is not Robinson (in which case  it certifies that $A$ is not Robinsonian). 
The complete algorithm is reported below.

\medskip
\begin{algorithm}[H]\label{alg:Robinson_recognition}
\caption{\textit{\textit{Robinson}}$(A)$}
\SetKwInput {KwIn}{input}
\SetKwInput {KwOut}{output}
\KwIn{a matrix $A \in \mathcal{S}^{n}$}
\KwOut{a Robinson ordering $\pi$ of $A$, or stating that $A$ is not Robinsonian}
\vspace{2ex}
$\sigma_0=\SFS(A)$\\
\For{$i=1 ,\dots n-2$}{
	$\sigma_i = \SFS_+(A, \sigma_{i-1})$\\
	\If{$\sigma_i$ is Robinson}{
		\Return: $\pi = \sigma_i$ 
	}
}
\Return: `$A$ is NOT Robinsonian'
\end{algorithm}
\medskip

As already mentioned earlier, the SFS algorithm applied to binary matrices reduces to Lex-BFS. 
As a warm-up we now show that our SFS  multisweep algorithm  terminates in  three  sweeps to recognize  whether a binary matrix $A$ is Robinsonian.
As a binary matrix $A$ is Robinsonian if and only if the corresponding graph is a unit interval graph  \cite{Roberts69}, this is coherent 
with the fact that one can recognize unit interval graphs  in three sweeps of Lex-BFS  
\cite[Thm~9]{Corneil04}. Hence we have a new proof for this result, which has similarities but yet differs  from the original proof in \cite{Corneil04}.

\begin{theorem}\label{thm: our proof for 3sweep algorithm uig}
Let $G$ be a  connected graph  and  let $A$ be its adjacency matrix. 
Consider  the orders   $\sigma_0=\SFS(A)$, $\sigma_1=\SFS_+(A, \sigma_0)$ and  $\sigma_2=\SFS_+(A, \sigma_1)$.
Then~$G$ is a unit interval graph (i.e., $A$ is Robinsonian) if and only if $\sigma_2$ is a Robinson ordering of $A$.
\end{theorem}

\begin{proof}
Clearly, if $\sigma_2$ is Robinson then $A$ is Robinsonian.
Assume now  that $A$ is Robinsonian; we show that  $\sigma_2$ is Robinson.
Suppose, for contradiction, that there exists a triple $x <_{\sigma_2} y <_{\sigma_2} z$ which is not Robinson, i.e., $A_{xz} > \min \{A_{xy}, A_{yz}\}$. 
Then the path $(x,z)$ avoids $y$ and thus, in view of Lemma \ref{thm:PAL} (PAL), in any Robinson ordering $\pi$ one cannot have $x <_{\pi} z <_{\pi} y$.
We may assume without loss of generality  that  $z <_{\pi} x <_{\pi} y$ in some Robinson ordering $\pi$.
Because $A$ is a binary matrix,
then $A_{xz}=1$, $A_{yz}=0$ and thus 
 $\{x,z\} \in E,  \{y,z\} \notin E$.
By construction, $\sigma_1$ is a good SFS ordering of $A$ starting (say) at the anchor $a$. Let $\mathcal{L}=\{L_0, L_1 \dots, L_r\}$ be the similarity layer structure of $A$ rooted at $a$.  By Lemma~\ref{thm:layers and SFS/Robinson ordering}, we know that $\mathcal L$ is compatible with $\sigma_1$, i.e., $a<_{\sigma_1} L_1 <_{\sigma_1} \ldots
<_{\sigma_1} L_r$. Using Theorem~\ref{thm:flipping layers theorem} 
we obtain that 
$L_r <_{\sigma_2}L_{r-1} <_{\sigma_2} \ldots <_{\sigma_2} L_1 <_{\sigma_2} a.$
Moreover,  using Lemma \ref{thm:inequalities among layers} and the fact that $G$ is connected, it is easy to see that each layer $L_i$ is a clique of $G$.
Hence, $y,z$ cannot be in the same layer of $\mathcal{L}$, as $\{y,z\} \notin E$.
Since  $y <_{\sigma_2} z$, it  follows that $z \in L_i, y \in L_j$ with $i < j$ and thus $z<_{\sigma_1} y$. 
Say $x\in L_h$. One cannot have  $h<j$ since this would contradict $x<_{\sigma_2} y$. 
If   $h=j$ then $x,y\in L_j$ and thus $A_{zx}=A_{zy}$ by definition of the layers, contradicting the fact that $A_{xz}=1$, $A_{yz}=0$.
Hence one must have $j<h$. Then $z\in L_i$, $y\in L_j$, $x\in L_h$ with $i<j<h$ and thus  $z<_{\sigma_1} y<_{\sigma_1} x$.
Now we get a contradiction with Lemma \ref{thm:PAL} (PAL), as $z<_\pi x<_\pi y$ and the path $(x,z)$ avoids~$y$.
\hfill
\end{proof}

The proof of Theorem~\ref{thm: our proof for 3sweep algorithm uig} outlines a fundamental difference between unit interval graphs and Robinsonian matrices.
Indeed, using Lemma~\ref{thm:inequalities among layers}, it is easy to see that, for $0/1$ Robinsonian matrices, each layer~$L_i$ of the similarity layer structure $\mathcal{L}$ rooted at an anchor $a$ is a clique of $G$.
This property  in fact  permits  to bound by three the number of sweeps neded to recognize $0/1$ Robinsonian matrices.
However,  for Robinsonian matrices with at least three distinct values we do not have any analogous structural property for the vertices lying in a common layer, which explains why we might need $n-1$ sweeps in the worst case.

We now formulate our main result, namely that the SFS multisweep algorithm terminates in at most $n-1$ steps to recognize whether an  $n \times n$ matrix is Robinsonian.

\begin{theorem}\label{thm:final theorem}
Let $A \in \MS^n$ and let $\sigma_0=\SFS(A)$, $\sigma_i=\SFS_+(A,\sigma_{i-1})$ for $i \geq 1$ be  the successive sweeps returned by Algorithm \ref{alg:Robinson_recognition}.
Then $A$ is a Robinsonian matrix if and only if $\sigma_{n-2}$ is a Robinson ordering of $A$.
\end{theorem}

We will  give  the full proof of Theorem \ref{thm:final theorem}  in Subsection \ref{sec:5-final proof} below.
What we need to show is that if $A$ is Robinsonian then the order $\sigma_{n-2}$ in Algorithm \ref{alg:Robinson_recognition} is a Robinson ordering of $A$. 
We now give a  rough sketch of the strategy which we will use to prove this result. 
The proof will  use induction on the size $n$ of the matrix $A$.

As was shown  earlier,  the sweep $\sigma_1$ is a  good SFS ordering of $A$  with  end-vertices  (say) $a$ and $b$, and all subsequent sweeps have  the same end-vertices (flipping their order at each sweep) in view of Theorem \ref{thm:flipping theorem}.
A first key ingredient will be to show that if we delete both end-vertices $a$ and $b$ and set $S=V\setminus \{a,b\}$, then the induced  order $\sigma_3[S]$ is  a good SFS ordering of the principal  submatrix $A[S]$. A second crucial ingredient will be to show that the induced order $\sigma_{n-2}[S]$ can be obtained with the multisweep algorithm applied to $A[S]$ starting from  
$\sigma_3[S]$. This will enable us to apply the induction  assumption and  to conclude that $\sigma_{n-2}[S]$ is a Robinson ordering of $A[S]$.
Hence all triples $(x,y,z)$ in $\sigma_{n-2}$ that are contained in $S=V\setminus \{a,b\}$ are Robinson. 
The last step is to show that all triples $(x,y,z)$ in $\sigma_{n-2}$ that contain $a$ or $b$ are also Robinson.

As we see in the above sketch, the sweep $\sigma_3$ plays a special role. It is obtained by applying three sweeps of $\SFS_+$ starting from the good SFS ordering $\sigma_1$. For this reason  we call it a {\em 3-good SFS ordering}. We  introduce and investigate in detail this 
 notion of `3-good SFS ordering' in Subsection \ref{sec:5-three-good SFS ordering} below.

\subsection{3-good SFS orderings} \label{sec:5-three-good SFS ordering}

Consider a Robinsonian matrix $A\in \mathcal S^n$. Recall that a SFS ordering $\tau$ of $A$ is said to be  {\em good}  if its first vertex is an anchor of $A$ (see Subsection \ref{sec:4-good SFS and end points}).
We now introduce the notion of 3-good SFS ordering. A linear order~$\tau$  is  called a {\em 3-good SFS  ordering} of $A$ 
if there exists a  good SFS ordering $\tau'$ of $A$ such that, if we set  $\tau''=\SFS_+(A,\tau')$, then $\tau=\SFS_+(A,\tau'')$ holds. 
In other words, a 3-good SFS ordering is obtained by performing three consecutive good sweeps.
Of course any 3-good SFS ordering is also a good SFS ordering. Furthermore, if we consider Algorithm \ref{alg:Robinson_recognition}, then any sweep $\sigma_i$ with $i \geq 3$ is a 3-good SFS ordering by construction. 
First we report the following flipping property of layers which follows 
as a direct application of Theorem \ref{thm:flipping layers theorem}.

\begin{corollary}\label{thm:flipping layers for 3-good SFS}
Assume  $A \in \MS^n$ is a Robinsonian matrix.
Let $\tau'$ be a good SFS ordering of $A$,
 $\tau''=\SFS_+(A,\tau')$ and $\tau=\SFS_+(A,\tau'')$.
Let $\mathcal{L}=\{L_0,\dots,L_r\}$ be the similarity layer structure of $A$ rooted at the first vertex of $\tau$.
If $x \in L_i$, $y\in L_j$ with $i<j$ then $y<_{\tau''} x$.
\end{corollary}

We now show some important properties of  3-good SFS orderings, 
that we will use in  the proof of correctness of the multisweep algorithm.
First we show that  some triples in a 3-good SFS ordering can  be shown to be Robinson.

\begin{lemma}\label{thm:3-good-sweep basic fact not-Robinson triple}
Assume  $A \in \MS^n$ is a Robinsonian matrix.
Let $\tau$ be a 3-good SFS ordering starting at $a$ and ending at $b$.
Let $\mathcal{L}=\{L_0=\{a\}, L_1,\dots,L_r \}$ be the similarity layer structure of $A$ rooted at $a$.
Then the following holds:
\begin{enumerate}[(i)]
\item If $x <_{\tau} y <_{\tau} z$ and $(x,y,z)$ is not Robinson, then $x,y,z \in L_i$ with $1 \leq i \leq r$.
\item Every triple $(a,x,y)$ with $x<_{\tau} y$ is Robinson.
\item Every triple $(x,y,b)$ with $x<_{\tau} y$ is Robinson.
\end{enumerate}
\end{lemma}

\begin{proof}
Let $\tau'$ be a good SFS order  such that  $\tau''=\SFS_+(A,\tau')$, $\tau=\SFS_+(A,\tau'')$. 
Let  $\mathcal{L}''=\{L''_0=\{b\}, L_1'',\dots \}$ denote the similarity layer structure of $A$ rooted at $b$, which is compatible with $\tau''$.

\textit{(i)}
Let  $x<_{\tau} y<_{\tau} z$ such that  $x,y,z$ do not all belong to the same layer of~$\mathcal{L}$ and assume that  $(x,y,z)$ is not Robinson.
Then $A_{xz} > \min \{A_{xy},A_{yz}\}$ and the path $(x,z)$ avoids $y$.
Let $\pi$ be a Robinson ordering and assume, without loss of generality, that $x <_{\pi} z$. 
Then, since $(x,z)$ avoids $y$, in view of Lemma \ref{thm:path avoiding and Robinson ordering} $y$ cannot appear between $x$ and $z$ in any Robinson ordering.
If $y$ appears after $z$ in $\pi$ then we have $x <_{\pi} z <_{\pi} y$ and $x <_{\tau} y <_{\tau} z$, and we get a contradiction with Lemma \ref{thm:PAL} (PAL) as there cannot exists a path from $x$ to $z$ avoiding $y$.
Therefore $y <_\pi x <_\pi z$ and thus $A_{xz}>\min\{A_{xy},A_{yz}\}= A_{yz}.$
In view of Lemma \ref{thm:layers and SFS/Robinson ordering},  $x,y,z$ do not belong to three distinct layers of $\mathcal L$ (since otherwise $(x,y,z)$ would be Robinson). Moreover, one cannot have $x\in L_i$ and $y,z\in L_j$ with $i<j$ (since this would imply $A_{xy}=A_{xz}\le A_{yz}$, a contradiction). 
Hence we must have  $x,y \in L_i$ and $z\in L_j$ with $i<j$.

Consider now $\tau''$;  applying Corollary \ref{thm:flipping layers for 3-good SFS}, we derive that $z <_{\tau''} x,y.$ 
Moreover, we cannot have that $z<_{\tau''} y <_{\tau''} x$, since we would get a contradiction with Lemma \ref{thm:PAL}~(PAL) as $z<_{\pi^{-1}} x <_{\pi^{-1}} y$ and the path $(x,z)$ avoids $y$.  Hence we have
$ z<_{\tau''} x <_{\tau''} y.$
Summarizing, the triple $(x,y,z)$ satisfies the properties:
\begin{equation}\label{eqxyz}
x,y\in L_i,\quad \ z\in L_j,\quad \ x<_{\tau} y<_{\tau} z, \quad \ x<_{\tau''} y, \quad\ y<_\pi x <_\pi z.
\end{equation}
We will now show that the properties in (\ref{eqxyz}) (together with the inequality $A_{xz}>A_{yz}$) permit to find an element $x_1<_{\tau}x$ for which the  triple $(x_1,y,z)$ again satisfies  the properties of (\ref{eqxyz}), replacing $x$ by $x_1$.
Then, iterating this construction leads to a contradiction.

We now proceed to show the existence of such an element $x_1$. As $x<_{\tau''} y$ and $x<_{\tau} y$,   $x,y$ are not  tied in $\tau$ and thus there exists $x_1<_{\tau} x$ such that $$A_{x_1x}>A_{x_1y}.$$ 
This implies  $x_1\in L_i$ (recall Lemma \ref{thm:inequalities among layers}).
Moreover, the path $(x_1,x,z)$ avoids $y$, since $A_{x_1x}>A_{x_1y}$ and $A_{xz}>A_{yz}.$
By construction we have:
$x_1<_{\tau} x <_{\tau} y <_{\tau} z.$
We claim that 
$$y <_\pi x_1 <_\pi z.$$
Indeed, if  $x_1<_\pi y$, then  $x_1<_\pi y<_\pi x$ and thus $A_{x_1x}\le A_{x_1y}$, a contradiction.
Moreover, if $z<_\pi x_1$ then $y<_\pi x<_\pi z<_\pi x_1$, which implies
$A_{x_1z}\ge A_{x_1x} > A_{x_1y}$ and thus the triple $(x_1,y,z)$ is not Robinson.
Then $A_{x_1z} > \min \{A_{x_1y},A_{yz}\}$ and the path $(x_1,z)$ avoids $y$.
Now, as $x_1<_{\tau} y <_{\tau} z$ and $x_1 <_{\pi^{-1}} z <_{\pi^{-1}} y$, we get a contradiction with Lemma \ref{thm:PAL} (PAL). 
So we have shown that $y <_\pi x_1 <_\pi z.$

Next we claim that $x_1<_{\tau''} y$. Indeed, if $y<_{\tau''} x_1$ then $z<_{\tau''} y<_{\tau''} x_1$, which together with $z<_{\pi^{-1}} x_1 <_{\pi^{-1}} y$ and the fact that the path $(x_1,x,z)$ avoids $y$,  contradicts Lemma~\ref{thm:PAL} (PAL).
Hence we have shown that the triple $(x_1,y,z)$ satisfies~(\ref{eqxyz}), which concludes the proof of \textit{(i)}.

\textit{(ii)}
follows directly from \textit{(i)}, since any triple containing  $a$ is not contained in a unique layer, and thus it must be Robinson.

\textit{(iii)}
Assume for contradiction that $(x,y,b)$ is not Robinson for some $x <_{\tau} y$, i.e., $A_{bx} > \min \{A_{by},A_{xy}\}$. Then the path $(b,x)$ avoids $y$.
If $\pi$ is a Robinson ordering ending at $b$ (which exists since $b$ is an anchor) then we must have $y<_\pi x<_\pi b$ because,
 in view of Lemma \ref{thm:path avoiding and Robinson ordering}, $y$ cannot appear between $x$ and $b$ in any Robinson ordering.
Hence, $A_{bx}>A_{by}$.
Since $\tau''$ is compatible with $\mathcal L''$ which is rooted at $b$,  we must have $b <_{\tau''} x <_{\tau''} y$ and  moreover  $x,y$ belong to distinct layers of $\mathcal L''$.
Thus  $x \in L''_i, y \in L''_j$ with $i < j$ which,  in view of Theorem \ref{thm:flipping layers theorem}, implies  $y <_{\tau} x$, a contradiction.
$\qquad$
\end{proof}

As a first direct application  of Lemma \ref{thm:3-good-sweep basic fact not-Robinson triple}(i), we can conclude that the multisweep algorithm terminates in at most four steps when applied to a matrix $A$ whose similarity layers rooted at the end-vertex of the first sweep $\sigma_0$ all have cardinality at most 2.

\medskip
Consider a 3-good SFS ordering  $\tau$  of a Robinsonian matrix  $A$ with end-vertices $a$ and $b$ and consider the induced order $\tau[S]$ of the submatrix $A[S]$ indexed by the subset $S=V\setminus\{a,b\}$. In the next lemmas we  show some properties of $\tau[S]$. First, we show that $\tau[S]$ is a good SFS ordering of $A[S]$ (Lemma \ref{thm:3-good-sweep without endpoints}). Second, we show that applying  the $\SFS_+$ algorithm to $\tau$ and then deleting $a$ and $b$ yields the same order as applying the $\SFS_+$ algorithm 
to the induced order $\tau[S]$ (Lemma \ref{thm:3-good-sweep and successive sweep}). These properties will be used in the induction step for the proof of correctness of the multisweep algorithm 
 in the next subsection. We start with showing a `flipping property' of  the second smallest  element of $\tau$.

\begin{lemma}\label{thm:2-good-sweep and successive sweep}
Assume  $A \in \MS^n$ is a Robinsonian matrix.
Let $\tau'$ be a good SFS ordering of $A$, $\tau''=\SFS_+(A,\tau')$ and $\tau=\SFS_+(A,\tau'')$. 
Let $a$ be the first vertex of $\tau$.
Then the successor $a_1$ of $a$ in $\tau$ is the predecessor of $a$ in $\tau''$. 
\end{lemma}

\begin{proof}
As before, $\mathcal{L}=\{L_0=\{a\}, L_1,\dots, L_r\}$ is the layer structure of $A$ rooted at~$a$, which is compatible with $\tau$.
The slice of $a$ in $\tau$ is precisely the first layer $L_1$ in~$\mathcal{ L}$. 
By definition, $a_1$ is the element of $L_1$ coming last in $\tau''$. 
By the flipping property in Corollary \ref{thm:flipping layers for 3-good SFS},
we know that the layer $L_1$ comes last but one in $\tau''$, just before the layer $L_0=\{a\}$. 
Then $a_1$ is the element of $L_1$ appearing last in $\tau''$, and thus it coincides with the predecessor of $a$ in $\tau''$.
$\qquad$
\end{proof}

\begin{lemma}\label{thm:3-good-sweep without endpoints}
Assume $A \in \MS^n$ is a Robinsonian matrix.
Let $\tau$ be a 3-good SFS ordering of $A$ with end-vertices $a$ and $b$ and set $S=V\setminus \{a,b\}$. 
Then $\tau [S] $ is a good SFS ordering of $A[S]$.
\end{lemma}

\begin{proof}
Say that $a$ is the first element of $\tau$ and that $b$ is its last element.
Let $\mathcal L=(L_0, L_1,\ldots, L_r)$ be the similarity layer structure rooted at $a$, which is compatible with $\tau$.
First we show that $\tau[S]$ is a SFS ordering of $A[S]$. For this consider elements $x,y,z \in S$ such that $A_{xz} > A_{xy}$. Then $(x,y,z)$ is not Robinson and thus $x,y,z \in L_{i}$ with $i \ge 1$ in view of Lemma \ref{thm:3-good-sweep basic fact not-Robinson triple}.
As $\tau$ is a SFS ordering, then in view of Theorem \ref{thm:SFS ordering characterization} there exists $u <_{\tau} x$ such that $A_{uy} > A_{uz}$. We have $u\ne a$ (since  $u = a$ would imply $A_{uy}=A_{uz}$)
and thus  $u \in  S$. This shows that $\tau[S]$ is a SFS ordering of $A[S]$.
Finally~$\tau[S]$ is good since, in view of Lemma \ref{thm:2-good-sweep and successive sweep}, it starts at $a_1$, the successor of $a$ in~$\tau$, which  is an anchor of $A[V\setminus \{a\}]$ (and thus also of $A[S]$) using Theorem~\ref{thm:last vertex of SFS is an anchor}.
\end{proof}

\begin{lemma}\label{thm:3-good-sweep and successive sweep}
Assume  $A \in \MS^n$ is a Robinsonian matrix.
Let  $\tau$ be a 3-good SFS ordering with end-vertices $a$ and $b$. 
Let $\tau_+=\SFS_+(A,\tau)$ and $S=V \setminus \{a,b\}$.
Then $\tau_+[S]=\SFS_+(A[S], \tau[S])$. 
\end{lemma}

\begin{proof}
Say $b$ is the first element of $\tau$ and $a$ be its last element. Then $a$ is the first element of $\tau_+$ and $b$ is its last element (Theorem \ref{thm:flipping theorem}).
Let consider the similarity layer structure $\mathcal{L}=(L_0=\{a\}, L_1, \ldots, L_r)$ of $A$ rooted at $a$, which is compatible with~$\tau_+$ (and thus we denote here by $\mathcal{L}_{+}$).

Set $\sigma=\SFS_+(A[S], \tau[S])$.
Let $a_1$ the predecessor of $a$ in $\tau$. 
As $\tau_+$ is clearly also a 3-good SFS ordering then, in view of Lemma~\ref{thm:2-good-sweep and successive sweep}, $a_1$ is the successor of $a$ in $\tau_+$ and thus both  $\tau_+[S]$ and $\sigma$ start  at $a_1$.
Assume that $\sigma$ and $\tau_+[S]$ agree on their first $p$ elements $a_1,\ldots,a_p$, but not at the next $(p+1)$th element. That is,
$\tau_+[S] = (a_1,\ldots, a_p, x, \ldots, y,\ldots)$, while $\sigma=(a_1,\ldots,a_p, y, \ldots, x, \ldots)$, where $x,y$ are distinct elements.
We distinguish three cases.

Assume first that  $x,y$ are tied in $\tau_+$ (and thus in $\sigma$ too).
Then one must have $y<_{\tau} x$ (to place $x$ before $y$ in $\tau_{+}[S]$) and $x<_{\tau} y$ (to place $y$ before $x$ in $\sigma$), a contradiction.
Assume now that $x,y$ are not tied in $\tau_+$, but they are tied in $\sigma$.
Then $A_{a x}>A_{a y}$.
Hence, since the similarity layer structure $\mathcal{L}$ of $A$ is rooted at $a$, then we have  $x\in L_j$, $y\in L_k$ for some $j<k$. This implies $y<_\tau x$ (by Corollary \ref{thm:flipping layers for 3-good SFS}) and thus, since $x,y$ are tied in $\sigma$, one would place $x$ before $y$ in $\sigma$, a contradiction.

Assume finally that $x,y$ are not tied in $\tau_+$ and also not in $\sigma$.
Let $a_j$ be the pivot splitting  $x$ and $y$ in $\sigma$ so that $A_{a_jy}>A_{a_jx}$, with $1 \leq j \leq p$.
We claim that  $a$ is the pivot splitting $x$ and $y$ in $\tau_+[S]$.
For this,  suppose that  $a_i$ is the pivot splitting $x$ and~$y$ in $\tau_+[S]$ for some $1 \leq i \leq p$, so that $A_{a_i x}>A_{a_i y}$ and $i\ne j$.
It is now easy to see that $i>j$ would imply $y<_{\tau_+} x$, while $i<j$  would imply $x<_\sigma y$, a contradiction in both cases.
Hence,  $a$ is the pivot splitting $x,y$ in $\tau_+[S]$ and thus $A_{a x}>A_{a y}$. Then, as $\mathcal{L}_{+}$ is the similarity layer structure of $A$ rooted at $a$,  $x$ and $y$ belong to distinct layers of~$\mathcal{L}^+$.  
Moreover, $a_j<_{\tau_+}  x <_{\tau_+}  y$ and  the triple $(a_j,x,y)$ is not Robinson. As~$\tau_+$ is a 3-good SFS ordering, we can apply  Lemma \ref{thm:3-good-sweep basic fact not-Robinson triple} and  conclude that $a_j,x,y$ must belong to a common layer of $\mathcal L$, contradicting the fact that $x,y$ belong to distinct layers of~$\mathcal L^+$.
\end{proof}

\subsection{Proof of correctness of the multisweep algorithm} \label{sec:5-final proof}

We can finally put all ingredients together and show the correctness of our multisweep algorithm.
We show the following result, which implies  directly  Theorem~\ref{thm:final theorem}.

\begin{theorem}\label{thm:good sweep converges in n-2}
Let $A \in \MS^n$ be a Robinsonian matrix, let $\tau_1$ be a good SFS ordering of $A$ and let  $\tau_i =\SFS_+(A,\tau_{i-1})$ for $i \geq 2$.
Then $\tau_{n-2}$ is a Robinson ordering of~$A$.
\end{theorem}

\begin{proof}
The proof is by induction on the size $n$ of $A$. For $n <3$ there is nothing to prove and for $n=3$ the result  holds trivially.
Hence, suppose $n \geq 4$. Then, by the induction assumption, we know that the following holds:
\begin{equation*}
\begin{array}{l}
\text{If }  \sigma_1 \text{ is a good SFS ordering of a Robinsonian matrix } A' \in \MS^k
\text{ with } k \leq n-1\\
\text{and  } \sigma_i = \text{SFS}_+(A',\sigma_{i-1}) \text{ for } i\ge 2
\text{, then } \sigma_{k-2}  \text{ is a Robinson ordering of } A'.
\end{array}
\end{equation*}
Suppose $\tau_1$ starts with $a$ and ends with $b$.
By Lemma \ref{thm:flipping theorem}, the end-vertices of any $\tau_i$ with $i \geq 2$ are $a$ and $b$ (flipped at every consecutive sweep).
For any $i\ge 3$, $\tau_i$ is a 3-good SFS ordering of $A$. 
Hence, setting $S=V\setminus \{a,b\}$, in view of  Lemma \ref{thm:3-good-sweep and successive sweep}, 
we obtain that $\tau_{i+1}[S] =\SFS_+(A[S],\tau_i[S])$ for each $i\ge 3$.
 
Consider the order $\sigma_1:=\tau_3[S]$ and the successive sweeps $\sigma_{i}=\SFS_+(A[S],\sigma_{i-1})$  ($i\ge 2$) returned by the multisweep algorithm applied to $A[S]$ starting from $\sigma_1$.

As $\tau_3$ is a 3-good SFS ordering of $A$, in view of Lemma~\ref{thm:3-good-sweep without endpoints} we know that  $\sigma_1$ is a good SFS ordering of $A[S]$.
Hence, using the induction assumption applied to $A[S]$ and $\sigma_1$, we can conclude that the sweep $\sigma_{|S|-2}=\sigma_{n-4}$ (returned by the multisweep algorithm applied to $A[S]$ with $\sigma_1$ as first sweep)
is a Robinson ordering of $A[S]$.

We now observe   that equality $\tau_{i+2}[S]=\sigma_i$ holds for all $i\ge 1$, using induction on $i\ge 1$. This is true for $i=1$ by the definition of $\sigma_1$.
Inductively, if $\tau_{i+2}[S]=\sigma_i$ then $\tau_{i+3}[S]= \SFS_+(A[S],\tau_{i+2}[S])= \SFS_+(A[S], \sigma_i)=\sigma_{i+1}$.
Hence, we can conclude that $\tau_{n-2}[S]=\sigma_{n-4}$ is a Robinson ordering of $A[S]$.

Finally, using  with Lemma \ref{thm:3-good-sweep basic fact not-Robinson triple}, we can conclude that all triples $(x,y,z)$ in $\tau_{n-2}$ that contain $a$ or $b$ are Robinson.
Therefore we have shown that $\tau_{n-2}$ is a Robinson ordering of $A$, which concludes the proof.
$\qquad$
\end{proof}

In other words, starting with a good SFS ordering of a Robinsonian matrix~$A\in~\mathcal S^n$, after at most $n-2$ sweeps we find a Robinson ordering of $A$.
Finally, we can now prove Theorem \ref{thm:final theorem}, since the last vertex of the first sweep $\sigma_0$ in Algorithm \ref{alg:Robinson_recognition} is an anchor of $A$ (Theorem \ref{thm:last vertex of SFS is an anchor}) and thus the second sweep $\sigma_1$ is a good SFS ordering.
Hence, if $A \in \MS^n$ is a Robinsonian matrix, in view of Theorem \ref{thm:good sweep converges in n-2}, the multisweep algorithm returns a Robinson ordering in at most $n-2$ sweeps starting from $\sigma_1$, and thus in at most $n-1$ sweeps counting also the initialization sweep $\sigma_0$.

\subsection{Worst case instances}\label{sec:5-worst case instances}

We present a class of $n \times n$ Robinsonian matrices, communicated to us by S. Tanigawa, for which the SFS multisweep algorithm (Algorithm~\ref{alg:Robinson_recognition}) needs $n-1$ sweeps to terminate.

\begin{definition}\label{def:worst matrix}
Let $A\in \MS^n$ be the Robinson matrix defined as follows:
\begin{align}
A_{1n}=0,\ A_{1i}=1 \quad & \text{ for } 2\le i\le n-1, \\
A_{2n}= 1,\ A_{in}=2 \quad & \text{ for } 3\le i\le n-1, \\
A_{ij}=A_{i-1,j+1} + 1 \quad & \text{ for } 2\le i<j\le n-1. \label{eqshift}
\end{align}
We will refer to (\ref{eqshift}) as the {\em shifting property}.
\end{definition}

We give below an example of such a matrix $A$ for $n=11$:

\begin{equation*}
A=
\bordermatrix{
~ & \textbf{1} & \textbf{2} & \textbf{3} & \textbf{4} & \textbf{5} & \textbf{6} & \textbf{7} & \textbf{8} & \textbf{9} & \! \! \textbf{10} & \!\!  \textbf{11}\cr
 \textbf{1}  & * & 1 & 1 & 1 & 1 & 1 & 1 & 1 & 1 & 1 & 0 \cr
 \textbf{2}  & & * & 2 & 2 & 2 & 2 & 2 & 2 & 2 & 1 & 1\cr
 \textbf{3}  & & & * & 3 & 3 & 3 & 3 & 3 & 2 & 2 & 2 \cr
 \textbf{4}  & & & & *  & 4 & 4 & 4 & 3 & 3 & 3 & 2\cr
 \textbf{5}  & & & & & * & 5 & 4 & 4 & 4 & 3 & 2\cr
 \textbf{6}  & & & & & & * & 5 & 5 & 4 & 3 & 2\cr
 \textbf{7}  & & & & & & & * & 5 & 4 & 3 & 2 \cr 
 \textbf{8}  & & & & & & & & * & 4 & 3 & 2\cr
 \textbf{9}  & & & & & & & & & *& 3 & 2 \cr
 \textbf{10}  & & & & & & & & & & * & 2\cr
 \textbf{11}  & & & & & & & & & & & * \cr
}.
\end{equation*}
Note that the matrix $A$ in Definition~\ref{def:worst matrix} is Robinson by construction, and therefore $(1,\dots,n)$ is a Robinson ordering of~$A$.
We consider the following ordering
\begin{equation}\label{eq:0}
\sigma_0=(2,3,\cdots,n,1),
\end{equation}
which can easily be checked to be a SFS ordering of $A$.  
We will consider  the SFS multisweep algorithm (Algorithm~\ref{alg:Robinson_recognition})
applied to the matrix~$A$ when taking  the ordering
$\sigma_0$ as initial ordering. 
As we show below,  the algorithm needs~$n-1$ sweeps to terminate.

\begin{theorem}\label{theoA}
Let $A$ be as in Definition~\ref{def:worst matrix},  let $\sigma_0=(2,\cdots,n,1)$ and let $\sigma_{i}=\SFS_+(A, \sigma_{i-1})$ for $1 \le i\le n-2$.
 Then the smallest index $j$ for which $\sigma_j$ is a Robinson ordering of $A$ is $j=n-2$. 
\end{theorem}

We first group  properties of the matrix $A$ 
needed for the proof of Theorem~\ref{theoA}.

\begin{lemma}
The following relations hold for the matrix  $A$ from Definition~\ref{def:worst matrix}:
\begin{align}
A_{j,j+1}=\cdots = A_{j,n-j} > A_{j,n-j+1}\ \   \text{ for  } 1\le j\le (n-1)/2, 
\label{eqP1}\\
A_{j,n-j+1}<A_{j+1,n-j+1}< A_{j+2,n-j+1} = \ignore{A_{j+3,n-j+1}=} \cdots = A_{n-j,n-j+1} \ \ \text{ for } 1\le j\le n/2, 
 \label{eqP4}\\
A_{j,n-j+3}>A_{j,n-j+4} \ \ \text{ for } 5\le j\le (n+2)/2. 
 \label{eqP5}
\end{align}
\end{lemma}
\begin{proof}
Relations (\ref{eqP1}) and (\ref{eqP4}) hold for $j=1$ 
 and  we use the shifting property~(\ref{eqshift}) to get the general case.
Analogously for relation (\ref{eqP5}), since it holds for~$j=5$.
$\qquad$
\end{proof}

As  key ingredient for proving Theorem \ref{theoA} we give the explicit description of the successive orderings $\sigma_1,\cdots, \sigma_{n-2}$ returned by  the multisweep algorithm.

\begin{lemma}\label{theosweep}
Let $A$ be as in Definition~\ref{def:worst matrix} and let $\sigma_0=(2,\cdots,n,1)$.  Then the  successive orderings $\sigma_{i}=\SFS_+(A,\sigma_{i-1})$  for $1\le i\le n-2$, returned by the multisweep algorithm applied to the matrix $A$,  have the  form
\begin{equation}\label{eqeven}
\sigma_{2k}= ((n,n-1,\cdots, n-k+1), (k+2,\cdots, n-k-1, n-k),(k+1,\cdots, 1)),
\end{equation}
for even order  $2k$ and the form 
\begin{equation}\label{eqodd}
\sigma_{2k+1}=((1,\cdots,k+1),(n-k-1,\cdots, k+2), (n-k,n-k+1,\cdots, n-1,n)),
\end{equation}
for odd order $2k+1$.
\end{lemma}

\begin{proof}
We show that the successive sweeps have the desired form using induction on the order  of the sweep.
In a first step,  let  $0\le k\le (n-3)/2$ and assume that $\sigma_{2k}$ has the form (\ref{eqeven}), i.e., 
$$\sigma_{2k}= (\underbrace{(n,n-1,\cdots, n-k+1)}_{\tau}, \underbrace{(k+2,\cdots, n-k-1)}_{\pi}, (n-k),\underbrace{(k+1,\cdots, 1)}_{\psi});$$
we show that $\sigma_{2k+1}=SFS_+(A,\sigma_{2k})$ has the form (\ref{eqodd}), i.e., 
$$\sigma_{2k+1}=(\underbrace{(1,\cdots,k+1)}_{\psi^{-1}},\underbrace{(n-k-1,\cdots, k+2)}_{\pi^{-1}}, (n-k),\underbrace{(n-k+1,\cdots, n)}_{\tau^{-1}}).$$
First we claim that for any $1\le j\le k+1$  the $j$th pivot is $p_j=j$, with corresponding  ordered partition:
$$((1,2,\cdots,j), \underbrace{\{j+1,\cdots, n-j\}, (n-j+1, \cdots, n)}_{\phi(j)}).$$
Recall from Section~\ref{sec:3-description algorithm} that, given a pivot $p_j$ at the current iteration of the SFS algorithm, we let~$\phi(p_j)$ denote the queue  of unvisited nodes induced by~$p_j$.
Hence, the set $\{j+1,\cdots, n-j\}$ represents the current slice, i.e., the first block of $\phi(p_j)$, whose elements are not yet ordered. 
The claim is  true for $j=1$ (easy to see). 
Assume this is true for $j\le k$, we show this also holds for $j+1$. Indeed the next pivot is $j+1$
(the vertex in the slice $ \{j+1,\cdots, n-j\}$ appearing last in $\sigma_{2k}$), 
which splits the vertex $n-j$ from the rest of the slice, leading to the new slice
$\{j+2,\cdots, n-j-1\}$ (since 
$A_{j+1,j+2}=\cdots = A_{j+1,n-j-1}>A_{j+1,n-j}$ by (\ref{eqP1})). Hence, the ordered partition becomes: 
$$((1,2,\cdots,j,j+1), \underbrace{\{j+2,\cdots, n-j-1\}, (n-j, \cdots, n)}_{\phi(j+1)}).$$
Hence, after the selection of the first $k+1$ pivots, the ordered partition is
$$((1,\cdots,k+1) \{k+2,\cdots, n-k-1\},(n-k,\cdots,n))=(\psi^{-1}, \{k+2,\cdots, n-k-1\}, (n-k), \tau^{-1}).$$
The next pivot is $n-k-1$ (the vertex in the slice appearing last in $\sigma_{2k}$).
 Using (\ref{eqP4}) (applied to $j=k+2$) we know that
$A_{n-k-1,k+2}<A_{n-k-1,k+3}<A_{n-k-1,k+4}=\cdots = A_{n-k-1,n-k-2}$, so that the ordered partition becomes
$$(\psi^{-1}, (n-k-1), \{k+4,\cdots, n-k-2\}, (k+3,k+2), (n-k), \tau^{-1}).$$
The next pivot is $n-k-2$. As  $A_{n-k-2,k+4}<A_{n-k-2,k+5}=\cdots=A_{n-k-2,n-k-3}$ 
(using again (\ref{eqP4}), now applied to $j=k+3$), the next ordered partition is
$$(\psi^{-1}, (n-k-1, n-k-2), \{k+5,\cdots, n-k-3\}, (k+4, k+3,k+2), (n-k), \tau^{-1}).$$
Iterating this process one can easily see that  the ordering returned by   the $\SFS_+$ algorithm has indeed the form (\ref{eqodd}).
\ignore{
Using~(\ref{eqP1}) and the tie-breaking rule of SFS$_{+}$, we yield the following similarity partition after the $k$-th pivot is chosen:
$$\left((1,2\cdots,k)\{k+1,\cdots,n-k\},(n-k+1,\cdots,n) \right) \ = \ (\psi^{-1}, \{k+1,\cdots,n-k\},  \tau^{-1}).$$
After $k$, we will choose as next pivot $n-k$, because it appear last in $\sigma_{2k}$ with respect to the other vertices in the slice $\{k+1,\cdots,n-k\}$. 
In view of~(\ref{eqP4}) applied to $j=k+1$, we have that 
$A_{n-k-1,n-k}=\cdots=A_{k+3,n-k}>A_{k+2,n-k}>A_{k+1,n-k}$. 
Hence when $n-k$ becomes a pivot, the elements $k+2$ and $k+1$ are split, with $k+2$ placed before $k+1$ in the queue of unvisited vertices. So the next ordered partition  is
$$(\psi^{-1}, (n-k), \{k+3,\cdots,n-k-1\},(k+2),(k+1),\tau^{-1}).$$
The next pivot is $n-k-1$, and $k+3$ is split from the slice, since by (\ref{eqP4}) applied to $j=k+2$ we have that 
$A_{n-k-2,n-k-1}=\cdots=A_{k+4,n-k-1}>A_{k+3,n-k-1}$. Hence we get next
$$(\psi^{-1}, (n-k),(n-k-1), \{k+4,\cdots,n-k-2\},(k+3), (k+2),(k+1),\tau^{-1}).$$
Iterating, at each next step we have that the last vertex in the slice becomes the next pivot, while the first vertex is split from the slice. 
Therefore we obtain that $\sigma_{2k+1}$ has indeed the form
$(\psi^{-1}, \pi^{-1}, (n-k+1), \tau^{-1})$.
}

In a second step, let $1\le k\le (n-2)/2$ and assume  $\sigma_{2k-1}$ is as  in (\ref{eqodd}) (after shifting indices), i.e., 
$$\sigma_{2k-1}=\underbrace{(1,\cdots,k)}_{\tau},\underbrace{ (n-k,\cdots, k+2)}_{\pi}, (k+1), \underbrace{(n-k+1,\cdots, n)}_{\psi}.$$
We show that $\SFS_+(A,\sigma_{2k-1})$ has the form (\ref{eqeven}), i.e., that it is equal to
$$\sigma_{2k}=(\underbrace{(n,\cdots,n-k+1)}_{\psi^{-1}}, \underbrace{(k+2,\cdots, n-k)}_{\pi^{-1}},(k+1), \underbrace{(k,\cdots, 1)}_{\tau^{-1}}.$$
First we claim that for $j\le k$ the $j$th pivot is $p_j=n-j+1$ with ordered  partition
$$((n,n-1,n-j+1), \underbrace{\{n-j,\cdots, j+2\}, (j+1,j,\cdots, 2,1)}_{\phi(n-j+1)}) .$$
This is true for $j=1$, because $n$  appears last in $\sigma_{2k-1}$ and $A_{n,1}<A_{n,2}<A_{n,3}=\cdots=A_{n,n-1}$. Assume this is true for some $j\le k-1$, we show this also holds for $j+1$. Indeed the next pivot is $n-j$. Moreover,  by (\ref{eqP4}), 
$A_{n-j,n-j-1}=\cdots =A_{n-j,j+3}>A_{n-j,j+2}$. 
Hence the new pivot $n-j$ splits the element $j+2$ from the rest of the slice, and the next ordered partition has the claimed form.
Hence, after $k$ steps, we have the following ordered partition:
$$(\psi^{-1},  \{k+2,\cdots, n-k\}, (k+1), \tau^{-1}).$$
Remains to show that the current slice $\{k+2,\cdots, n-k\}$ gets reordered as $\pi^{-1}$ in the next steps.
The next pivot is $k+2$.
By (\ref{eqP1}) (applied to $j=k+2$), $A_{k+2,k+3}= A_{k+2,k+4}=\cdots= A_{k+2,n-k-2} > A_{k+2,n-k-1}\ge A_{k+2,n-k}$.
Hence the two elements $n-k-1$ and $n-k$ are split by $k+2$ but as we cannot yet decide on their relative order they are both  placed in the same block  after the new slice  in the queue of unvisited vertices.
(Note indeed that, e.g.,  if $k=1$, then $A_{3,n-1}=A_{3,n-2}=2$.)
So we get the ordered partition:
$$ (\psi^{-1}, (k+2), \{k+3,\cdots, n-k-2\}, \{n-k-1,n-k\}, (k+1), \tau^{-1}).$$
The next pivot is $k+3$.
By (\ref{eqP1}) (applied to $j=k+3$),  $A_{k+3,k+4}=\cdots= A_{k+3,n-k-3}>A_{k+3,n-k-2}$. 
Hence  $k+3$ splits $n-k-2$ from the rest of the slice and the next ordered partition is
$$ (\psi^{-1}, (k+2,k+3),\{k+4,\cdots,n-k-3\}, (n-k-2) \{n-k-1,n-k\},\tau^{-1}).$$
The next pivot is $k+4$, which splits $n-k-3$ from the rest of the slice (using again~(\ref{eqP1})).
Moreover, by (\ref{eqP5}) (applied to $j=k+4$), $A_{k+4,n-k-1}>A_{k+4,n-k}$ and thus 
 the two elements $n-k-1$ and $n-k$ get ordered with $n-k-1$ coming before~$n-k$. So the ordered partition becomes
 $$ (\psi^{-1}, (k+2, k+3, k+4),\{k+5,\cdots,n-k-4\}, (n-k-3,n-k-2,n-k-1,n-k),\tau^{-1}).$$
Iterating one can easily conclude that the final ordering returned by the $\SFS_+$ algorithm indeed has the form (\ref{eqeven}). 
\end{proof}

We can now prove Theorem \ref{theoA}.

\begin{proof} {\em of Theorem \ref{theoA}.}
Using Lemma~\ref{theosweep}, we find that
$\sigma_{n-2}=(n, n-1,\cdots,1)$ for even $n$ and 
$\sigma_{n-2}=(1,2,\cdots,n)$ for odd $n$. Hence 
$\sigma_{n-2}$ is a Robinson ordering of~$A$ in both cases. 
Furthermore, observe that  $\sigma_{2k}\ne \sigma_{2k-1}^{-1}$ if $2k-1\le n-3$ (because~$n-k$ comes before $k+1$ in both $\sigma_{2k-1}$ and $\sigma_{2k}$), and $\sigma_{2k+1}\ne \sigma_{2k}^{-1}$ if $2k\le n-3$ (because $k+2$ comes before $n-k$ in both 
$\sigma_{2k+1}$ and $\sigma_{2k}$). 
Therefore, $\sigma_1,\cdots, \sigma_{n-3}$ cannot be Robinson orderings of $A$  in view of Lemma~\ref{thm:SFS+(a, pi) and reversed Robinson triple}.
This implies that the first index $j$ for which~$\sigma_j$ is Robinson is indeed $j=n-2$, which concludes the proof.
\end{proof}

It is important to remark that, for the class of matrices from Definition~\ref{def:worst matrix}, the fact that $n-1$ sweeps are required depends strongly on the choice of the initial ordering~$\sigma_0$ from (\ref{eq:0}).

\section{Complexity} \label{sec:6-complexity}

In this section we discuss the complexity of the SFS algorithm. 
Throughout we  assume that $A \in \MS^n$ is a nonnegative symmetric matrix,   given as  adjacency list of an undirected weighted graph $G=(V=[n],E)$. So $G$ is the support graph of $A$, whose edges are the pairs $\{x,y\}$ such that  $A_{xy}>0$ with edge weight $A_{xy}$, and $N(x)=\{y\in V:A_{xy}>0\}$ is the neighborhood of $x\in V$.
We assume that  each vertex $x \in V=[n]$ is linked to the list of vertex/weight pairs $(y,A_{xy})$ for its  neighbors $y\in N(x)$ and 
we let $m$ denote the number of nonzero entries of $A$.

\medskip
\begin{theorem}\label{thm:complexity_SFS}
The SFS algorithm (Algorithm \ref{alg:SFS}) applied to an $n\times n$  symmetric  nonnegative matrix  with $m$ nonzero entries runs in $O(n+m \log n)$ time.
\end{theorem}

\begin{proof}
As in \cite{Corneil04} For Lex-BFS, we may assume that we are given an initial order~$\tau$ of $V$ and that the vertices and their neighborhoods  are ordered according to  $\tau$ (in increasing order). This assumption is useful also for the discussion of the implementation of $\SFS_+$.

In order to run Algorithm \ref{alg:SFS}, we need to update the queue $\phi$ consisting of the unvisited vertices  at each iteration. 
The update consists in computing the similarity partition $\psi_p$  with respect to the current pivot $p$ and then refining $\phi$ by $\psi_p$. 

To maintain the priority among the unvisited vertices, the queue $\phi=(B_1,\dots,B_p)$ is stored in a linked list, whose elements are the classes  $B_1,\ldots,B_p$.
Moreover each vertex has a pointer to the class $B_i$ containing it and a pointer to its position in the class, which are updated throughout the algorithm.
This data structure permits constant time insertion and deletion of a vertex in $\phi$.

Initially,  the queue $\phi$ has only one class, namely the full set $V$.
At an  iteration of  Algorithm \ref{alg:SFS}, there are three main tasks to be performed: choose the next pivot, compute the similarity partition $\psi_p$ and refine the queue $\phi$ by $\psi_p$. 
\begin{enumerate}[(1)]
\item Choose the new pivot $p=p_i$. Since in Algorithm \ref{alg:SFS} the choice of the new pivot is arbitrary in case of ties, we will choose the first vertex of the first block in $\phi$. This operation can be done in constant time.
We then remove $p$ from the queue $\phi$  of unvisited vertices  and we update the queue $\phi$ by deleting $p$ from the class~$B_1$.

\item Compute the similarity partition $\psi_p=(C_1,\ldots,C_s)$ of the set $N_\phi(p)$
with respect to $p=p_i$ (as defined in Definition \ref{def:layer partition}).
Here $N_\phi(p)=N(p)\cap \phi$ denotes the set of unvisited vertices in the neighborhood $N(p)$ of $p$.
First we order the vertices~$y$ in~$N_\phi(p)$ for nonincreasing values of their similarities $A_{py}$ with respect  to $p$,
which  can be done in  in $O(|N_\phi(p)| \log |N_\phi(p)|)$ time using a sorting algorithm. 
Then we  create the similarity partition $\psi_p=(C_1,\ldots,C_s)$  simply by passing through the  elements in $N_\phi(p)$  in the order of nonincreasing similarities to $p$ which has just been found.
This task can be done in $O(|N_\phi(p)|)$ time. 
Finally we order the elements in each class $C_j$ (increasingly) according to $\tau$, which can be done in $O(|N_\phi(p)| \log |N_\phi(p)|)$.
So we have constructed  the ordered partition $\psi_p=(C_1,\dots,C_s)$ of $N_{\phi}(p)$ as a linked list, where  all classes  of $\psi_p$ are ordered according to $\tau$.
To conclude, the overall complexity of this second task is bounded by $O(|N_\phi(p)| \log |N_\phi(p)|)$.

\item The last task is to refine $\phi=(B_1,\dots, B_p)$ by $\psi_p=(C_1,\dots,C_s)$ (as defined in Definition \ref{def:refine}).
In order to obtain the new queue of unvisited vertices we proceed as follows:  
starting from $j=1$, for each class  $C_j$ of $\psi$,  we simply remove each vertex of $C_j$ from its corresponding class (say) $B_i$  in $\phi$  and we place it in a new class  $B'_i$ which we position immediately before $B_i$ in $\phi$.
Since both $C_j$ and $B_i$ are ordered according to $\tau$, the initial order $\tau$ in the new block $B'_i$ is preserved.
Using the above described data structure, such tasks can be performed in $O(|C_j|)$. Once a vertex is relocated in $\phi$, its pointers to the corresponding block and position in $\phi$ are updated accordingly.
Hence this last task  can be performed in time $O(\sum_{j=1}^{s}{|C_j|})=O(|N_\phi(p)|)$. 
\end{enumerate}

\smallskip Recall that at iteration $i$ we set $p=p_i$.
Then the complexity at the $i$th  iteration is  $O(1+|N_\phi(p_i)| \log |N_\phi(p_i)|)$.
Since we repeat the above three tasks for each vertex, then the  overall complexity  of Algorithm \ref{alg:SFS} is 
$O(\sum_{i=1}^n{(1+N_\phi(p_i)| \log |N_\phi(p_i)|}))=
O(n+m \log n)$.
$\qquad$
\end{proof}

Using the same data structure as above, we can show that  the SFS$_+$ algorithm can be implemented in the same running time as the SFS algorithm.
In fact, the only difference between the SFS algorithm and the $\SFS_+$ algorithm lies in the tie-breaking rule.
In the $\SFS_+$ algorithm, in case of ties we choose as next pivot the vertex in the slice appearing last in the given order $\sigma$.
Recall we assumed $V$ to be initially ordered according to a given linear order~$\tau$, which can be easily done in linear time in the size of the graph.
Then, we showed in the proof of Theorem~\ref{thm:complexity_SFS} that the initial order $\tau$ is always preserved in the classes of $\phi$ throughout the algorithm. Therefore, if we choose $\tau = \sigma^{-1}$, then the first vertex in each slice $S$ is exactly the vertex of $S$ appearing last in $\sigma$. 

\begin{corollary}\label{thm:complexity SFS_+}
The $\SFS_+$ algorithm 
\ignore{(Algorithm \ref{alg:SFS+})} 
applied to an $n\times n$  symmetric nonnegative matrix with $m$ nonzero entries runs in $O(n+m \log n)$ time.
\end{corollary}

\smallskip
It follows directly from Corollary  \ref{thm:complexity SFS_+} that any SFS multisweep algorithm with $k$ sweeps  can be implemented in 
$O(k(n+m \log n))$. Indeed the only additional tasks we need to do are the following: when we start a new $\SFS_+$ sweep we need to  reorder the vertices and their neighborhoods according to the reversal of the previous sweep,  and we need to check if the current sweep $\sigma_i$ is a Robinson ordering, which can  both be done in $O(m + n)$ time.

Alternatively, one can check at each iteration $i \geq 1$ if $\sigma_i = \sigma_{i-1}^{-1}$ holds, which requires only $O(n)$ time.
Then, if this is the case, in view of Corollary~\ref{thm:corollary alternative Robinson check} we know that $A$ is Robinsonian if and only $\sigma_i$ is a Robinson ordering.
Hence, one only needs to check once whether a given sweep is a Robinson ordering,
namely when it  is the reversed of the previous one.

Another similar approach is inspired by the one used in~\cite{Dusart15} for their multisweep algorithm to recognize cocomparability graphs.
Specifically, every time we compute an SFS ordering $\sigma_i$ with $i \geq 2$, we check if $\sigma_i = \sigma_{i-2}$. 
If this is the case, then we stop (because the algorithm will loop between $\sigma_{i-2}$ and $\sigma_{i-1}$) and we check whether $\sigma_i$ is Robinson.

In any case, as the multisweep algorithm (Algorithm \ref{alg:Robinson_recognition}) needs $k\le n-1$ sweeps, it runs in time $O(n^2+nm \log n)$.

\medskip
As already mentioned in Section \ref{sec:3-SFS algorithm}, if the matrix has only $0/1$ entries, then there is no need to order the neighborhood $N(p)$ of a given pivot $p$, because the similarity partition $\psi_{p}$ has only one class, equal to $N(p)$.
For this reason, 
in this case  the SFS algorithm can be implemented in linear time $O(m+n)$.
Furthermore, as shown in Theorem \ref{thm: our proof for 3sweep algorithm uig}, three sweeps suffice in the multisweep algorithm to find a Robinson ordering.
Therefore,  if $A$ is a binary matrix, the multisweep algorithm in~Algorithm \ref{alg:Robinson_recognition} has an overall running time of $O(m+n)$.
This is coherent with the fact that in the~$0/1$ case SFS reduces to  Lex-BFS. 

\medskip
When the graph $G$ associated to the matrix $A$ is connected the complexity of SFS and $\SFS_+$  is  $O(m\log n)$. Of course we may assume without loss of generality 
that we are in the connected case since we may deal with the connected components independently.
Indeed a matrix $A$ is Robinsonian if and only if the submatrices $A[C]$ are Robinsonian for all connected components $C$ of $G$, and Robinson orderings of the connected components $A[C]$ can be concatenated to give a Robinson ordering of the full matrix $A$.

\medskip
Finally we  observe that we may also exploit the potential  sparsity induced by  the {\em largest} entries of $A$.
While  $G$ is the graph whose edges are the pairs $\{x,y\}$ with entry $A_{xy}>0$ (where 0 is the smallest possible entry as $A$ is assumed to be nonnegative), we can
 also consider the graph $G'$ whose edges are the pairs $\{x,y\}$ with entry $A_{xy}<A_{\max}$, where $A_{\max}$ is the largest possible entry of $A$.
 Let $N'(p)$ denote the neighborhood of a vertex $p$ in $G'$ and 
let $m'$ denote the number of entries with $A_{xy}<A_{\max}$.
We claim that the SFS ($\SFS_+$) algorithm can also be implemented in time $O(n+m'\log n)$.

For this we  modify the definition of the similarity partition of a vertex $p$, which is now a partition of $N'(p)$ (so that the vertices $y\not\in N'(p)$ have entry $A_{py}=A_{\max}$) and the refinement of the queue $\phi$ by it:
while we previously build the queue $\phi$ of unvisited vertices using  a `push-first' strategy (put the vertices with highest similarity first) we now build the queue with a `push-last' strategy (put the vertices with lowest similarity last).

\section{Conclusions} \label{sec:7-conclusions}

In this paper we have introduced the new search algorithm \textit{Similarity-First Search} (SFS) and its variant $\SFS_+$, which are generalizations to weighted graphs  of the classical Lex-BFS algorithm and its variant Lex-BFS$_+$.
The algorithm is entirely based on the main task of partition refinement, it is conceptually simple and easy to implement. 
We have  shown that a multisweep algorithm can be designed using  SFS and $\SFS_+$, which permits to recognize if a symmetric $n\times n$ matrix  is Robinsonian and if so to return  a Robinson ordering after at most $n-1$ sweeps. We believe that this recognition algorithm is substantially simpler than the other existing algorithms. 
Moreover, to the best of our knowledge, this is the first work extending multisweep graph search algorithms to the setting of weighted graphs (i.e., matrices). 

\medskip
Our algorithm can also be used to recognize Robinsonian dissimilarities.
Recall that a  matrix $D\in \mathcal S^n$ is  a {\em Robinson dissimilarity matrix}  if 
$D_{xz} \geq \max\{D_{xy},D_{yz}\} $ for all $1\le x < y < z \le n$, and 
 a {\em Robinsonian dissimilarity} if its rows and columns can be simultaneously reordered to get a Robinson dissimilarity matrix. Clearly $D$ is a Robinsonian dissimilarity matrix if and only if the matrix $A=-D$ is a Robinsonian similarity matrix.
Therefore, one can check whether $D$ is a Robinsonian dissimilarity by applying the SFS-based multisweep algorithm to the matrix $A$. 

Alternatively one may also modify the SFS algorithm so that it can deal directly with dissimilarity matrices.
Say $D$ is a nonnegative dissimilarity matrix and $G$ is the corresponding weighted graph with edges the pairs $\{x,y\}$ with $D_{xy}>0$.
Then we can modify the SFS algorithm as follows. First, we now order the vertices in the neighborhood $N(p)$ of a vertex $p$ for {\em nondecreasing values} of the dissimilarities $D_{py}$ (instead of nonincreasing values of the similarities $A_{py}$ as was the case in SFS). Then we  construct the (dis)similarity partition $\psi_p$ of $N(p)$ by grouping the vertices with the same dissimilarity to $p$, in increasing values of the dissimilarities. 
Finally, when refining the queue $\phi$ by $\psi_p$, we apply a `push-first' strategy and place the vertices with lowest dissimilarity first. 
The resulting algorithm, which we name $\DiSFS$, standing for {\em Dissimilarity-Search First}, has the same running time in $O(n+m\log n)$. 
Moreover, as explained above at the end of Section \ref{sec:6-complexity}, it can also be implemented in time $O(n+m'\log n)$, where $m'$ denotes the number of entries of $D$ satisfying $D_{xy}<D_{\max}$ and  $D_{\max}$ denotes the largest entry of $D$.
Using $\DiSFS$ we can define the analogous multisweep algorithm for recognizing Robinsonian dissimilarities in time $O(n^2+nm\log n)$ (or $O(n^2+nm'\log n)$).

\medskip
As we have seen in Subsection \ref{sec:5-worst case instances}, there exists a family of $n \times n$ Robinsonian matrices where $n-1$ sweeps are  needed. 
It is an open question whether the multisweep algorithm can be modified in such a way  that it would need only a constant number of sweeps, in which case it might become competitive with the optimal algorithm of \cite{Prea14}.
For this one would  need to define another variant of SFS. 
A possible variant is  when  ties are broken using the SFS orderings returned by two previous sweeps (instead of only one as in the $\SFS_+$ variant). 
This approach has been succesfully applied to Lex-BFS in \cite{Corneil09} for the recognition of interval graphs in five Lex-BFS sweeps; there the last sweep used is the variant Lex-BFS$_*$,  which breaks ties using the linear order returned by two previous sweeps.
Dusart and Habib \cite{Dusart15} conjecture that a similar approach applies to recognize  cocomparability graphs with a constant number of sweeps.
Investigating  whether such an approach applies to Robinsonian matrices will be the subject of future work.
 
\medskip   
Finally, it will be interesting to investigate whether the new SFS algorithm can be used to study other classes of structured matrices and in the general area of similarity search and clustering analysis.

\section*{Acknowledgements}

This work was supported by the Marie Curie Initial Training Network ``Mixed Integer Nonlinear Optimization" (MINO)  grant no. 316647.
We are grateful to the  anonymous referees for their useful comments and to Shinichi Tanigawa for communicating to us  the  family of instances presented in Subsection~\ref{sec:5-worst case instances}, allowing us to include it in the paper.
The second author also thanks Michel Habib for his useful comment about the complexity of the algorithm and Mohammed El Kebir for useful discussions.


\begin{thebibliography}{10}

\bibitem{Atkins98}
J.E. Atkins, E.G. Boman, and B.~Hendrickson.
\newblock A spectral algorithm for seriation and the consecutive ones problem.
\newblock {\em SIAM Journal on Computing}, 28:297--310, 1998.

\bibitem{Barnard93}
S.T. Barnard, A.~Pothen, and H.D. Simon.
\newblock A spectral algorithm for envelope reduction of sparse matrices.
\newblock In {\em Proceedings of the 1993 ACM/IEEE Conference on
  Supercomputing}, Supercomputing '93, pages 493--502, 1993.

\bibitem{Booth76}
K.S. Booth and G.S. Lueker.
\newblock Testing for the consecutive ones property, interval graphs, and graph
  planarity using {PQ}-tree algorithms.
\newblock {\em Journal of Computer and System Sciences}, 13(3):335--379, 1976.

\bibitem{Bretscher08}
A.~Bretscher, D.G. Corneil, M.~Habib, and C.~Paul.
\newblock A simple linear time {LexBFS} cograph recognition algorithm.
\newblock {\em SIAM Journal on Discrete Mathematics}, 22:1277--1296, 2008.

\bibitem{Chepoi97}
V.~Chepoi and B.~Fichet.
\newblock Recognition of {R}obinsonian dissimilarities.
\newblock {\em Journal of Classification}, 14(2):311--325, 1997.

\bibitem{Chepoi09}
V.~Chepoi, B.~Fichet, and M.~Seston.
\newblock Seriation in the presence of errors: Np-hardness of
  {$l_\infty$}-fitting {R}obinson structures to dissimilarity matrices.
\newblock {\em Journal of Classification}, 26(3):279--296, 2009.

\bibitem{Chepoi11}
V.~Chepoi and M.~Seston.
\newblock Seriation in the presence of errors: A factor 16 approximation
  algorithm for {$l_\infty$}-fitting {R}obinson structures to distances.
\newblock {\em Algorithmica}, 59(4):521--568, 2011.

\bibitem{Corneil95}
D.G. Corneil, H.~Kim, S.~Natarajan, S.~Olariu, and A.P. Sprague.
\newblock Simple linear time recognition of unit interval graphs.
\newblock {\em Information Processing Letters}, 55(2):99--104, 1995.

\bibitem{Corneil04}
D.G. Corneil.
\newblock A simple 3-sweep {LBFS} algorithm for the recognition of unit
  interval graphs.
\newblock {\em Discrete Applied Mathematics}, 138(3):371--379, 2004.

\bibitem{Corneil05}
D.G. Corneil.
\newblock Lexicographic breadth first search –- {A} survey.
\newblock In J.~Hromkovič, M.~Nagl, and B.~Westfechtel, editors, {\em
  Graph-Theoretic Concepts in Computer Science}, volume 3353 of {\em Lecture
  Notes in Computer Science}, pages 1--19. Springer Berlin Heidelberg, 2005.

\bibitem{Corneil09}
D.G. Corneil, S.~Olariu, and L.~Stewart.
\newblock The {LBFS} structure and recognition of interval graphs.
\newblock {\em SIAM Journal on Discrete Mathematics}, 23(4):1905--1953, 2009.

\bibitem{Corneil10}
D.G. Corneil, E.~Köhler, and J.M. Lanlignel.
\newblock On end-vertices of lexicographic breadth first searches.
\newblock {\em Discrete Applied Mathematics}, 158(5):434 -- 443, 2010.

\bibitem{Corneil16}
{\sc D.G. Corneil, J.Dusart, M.~Habib, A.~Mamcarz, and F.~de~Montgolfier}, {\em
  A tie-break model for graph search}, Discrete Applied Mathematics, 199
  (2016), pp.~89--100.
\newblock Sixth Workshop on Graph Classes, Optimization, and Width Parameters,
  Santorini, Greece, October 2013.

\bibitem{Habib13}
P.~Crescenzi, D.G. Corneil, J.~Dusart, and M.~Habib.
\newblock New trends for graph search.
\newblock PRIMA Conference in Shanghai, June 2013.
\newblock available at
  \url{http://math.sjtu.edu.cn/conference/Bannai/2013/data/20130629B/slides1.pdf}.

\bibitem{Dusart15}
J.~Dusart and M.~Habib.
\newblock A new {LBFS}-based algorithm for cocomparability graph recognition.
\newblock {\em Discrete Applied Mathematics}, to appear, 2015.

\bibitem{Fogel14}
F.~Fogel, A.~d'~Aspremont, and M.~Vojnovic.
\newblock Serialrank: Spectral ranking using seriation.
\newblock In {\em Advances in Neural Information Processing Systems 27}, pages
  900--908. 2014.

\bibitem{Fogel15}
F.~Fogel, R.~Jenatton, F.~Bach, and A.~d'Aspremont.
\newblock Convex relaxations for permutation problems.
\newblock {\em SIAM Journal on Matrix Analysis and Applications},
  36(4):1465--1488, 2015.

\bibitem{Goulermas15}
J.Y. Goulermas, A.~Kostopoulos, and T.Mu.
\newblock A new measure for analyzing and fusing sequences of objects.
\newblock {\em IEEE Transactions on Pattern Analysis and Machine Intelligence},
  PP(99), 2015.

\bibitem{Habib00}
M.~Habib, R.~McConnell, C.~Paul, and L.~Viennot.
\newblock {Lex-BFS} and partition refinement, with applications to transitive
  orientation, interval graph recognition and consecutive ones testing.
\newblock {\em Theoretical Computer Science}, 234(1–2):59--84, 2000.

\bibitem{Hahsler08}
M.~Hahsler, K.~Hornik, and C.~Buchta.
\newblock Getting things in order: An introduction to the {R} package
  seriation.
\newblock {\em Journal of Statistical Software}, 25(1), 2008.

\bibitem{Laurent15}
M.~Laurent and M.~Seminaroti.
\newblock A {Lex-BFS}-based recognition algorithm for {R}obinsonian matrices.
\newblock In {\em Algorithms and Complexity: Proceedings of the 9th
  International Conference CIAC 2015}, volume 9079 of {\em Lecture Notes in
  Computer Science}, pages 325--338. Springer-Verlag, 2015.

\bibitem{Innar10}
I.~Liiv.
\newblock Seriation and matrix reordering methods: An historical overview.
\newblock {\em Statistical Analysis and Data Mining}, 3(2):70--91, 2010.

\bibitem{Looges93}
P.J. Looges and S.~Olariu.
\newblock Optimal greedy algorithms for indifference graphs.
\newblock {\em Computers \& Mathematics with Applications}, 25(7):15--25, 1993.

\bibitem{Mirkin84}
B.G. Mirkin and S.N. Rodin.
\newblock {\em Graphs and Genes}.
\newblock Biomathematics (Springer-Verlag). Springer, 1984.

\bibitem{Paige87}
R.~Paige and R.E. Tarjan.
\newblock Three partition refinement algorithms.
\newblock {\em SIAM Journal on Computing}, 16(6):973--989, 1987.

\bibitem{Prea14}
P.~Pr\'ea and D.~Fortin.
\newblock An optimal algorithm to recognize {R}obinsonian dissimilarities.
\newblock {\em Journal of Classification}, 31:1--35, 2014.

\bibitem{Roberts69}
F.S. Roberts.
\newblock Indifference graphs.
\newblock In {\em Proof Techniques in Graph Theory: Proceedings of the Second
  Ann Arbor Graph Theory Conference, F. Harary, ed.}, pages 139--–146, New
  York, NY, 1969. Academic Press.

\bibitem{Robinson51}
W.S. Robinson.
\newblock A method for chronologically ordering archaeological deposits.
\newblock {\em American Antiquity}, 16(4):293--301, 1951.

\bibitem{Rose75}
D.J. Rose and R.E. Tarjan.
\newblock Algorithmic aspects of vertex elimination.
\newblock In {\em Proceedings of 7th Annual ACM Symposium on Theory of
  Computing}, STOC '75, pages 245--254. ACM, 1975.

\bibitem{Seston08}
M.~Seston.
\newblock {\em Dissimilarit\'es de {R}obinson: Algorithmes de Reconnaissance et
  d'Approximation}.
\newblock PhD thesis, Universit\'e de la M\'editerran\'ee, 2008.

\bibitem{Simon91}
K.~Simon.
\newblock A new simple linear algorithm to recognize interval graphs.
\newblock In H.~Bieri and H.~Noltemeier, editors, {\em Computational
  geometry-methods, algorithms and applications}, volume 553 of {\em Lecture
  Notes in Computer Science}, pages 289--308. Springer Berlin Heidelberg, 1991.

\end{thebibliography}
\bibliographystyle{plain}

\end{document}